\documentclass[pdflatex,sn-mathphys-num]{sn-jnl}

\usepackage{graphicx}%
\usepackage{multirow}%
\usepackage{amsmath,amssymb,amsfonts}%
\usepackage{amsthm}%
\usepackage{mathrsfs}%
\usepackage[title]{appendix}%
\usepackage{xcolor}%
\usepackage{textcomp}%
\usepackage{manyfoot}%
\usepackage{booktabs}%
\usepackage{algorithm}%
\usepackage{algorithmicx}%
\usepackage{algpseudocode}%
\usepackage{listings}%

\usepackage{enumitem}
\usepackage{tikz}
\usetikzlibrary{positioning, shapes, arrows.meta, fit}
\usepackage{framed} 
\usepackage{mdframed}
\usetikzlibrary{decorations.pathreplacing}
\usetikzlibrary{decorations.markings}
\tikzstyle{vertex}=[circle, draw, inner sep=0pt, minimum size=3pt]
\newcommand{\vertex}{\node[vertex]}

\theoremstyle{thmstyleone}%
\newtheorem{thm}{Theorem}

\newtheorem{lem}[thm]{Lemma}%

\newtheorem{cor}[thm]{Corollary}
\newtheorem{claim}{Claim}[section]

\theoremstyle{thmstyletwo}%
\newtheorem{rem}{Remark}%

\theoremstyle{thmstylethree}%
\newtheorem{defn}{Definition}%
\theoremstyle{thmstylethree}
\newtheorem{redn}{Reduction Rule}
 
\newenvironment{claimproof}[1][Proof of claim]{%
  \begin{proof}[#1]%
}{%
  \end{proof}%
}

\begin{document}

\title[Article Title]{Hardness and Tractability of $\mathcal{T}_{h+1}$-\textnormal{\textsc{Free Edge Deletion}} \footnote[3]{A preliminary version has been published in the proceedings of the 25th International Symposium on Fundamentals of Computation Theory
(FCT) 2025.}}


\author[1]{\fnm{Ajinkya} \sur{Gaikwad}}\email{ajinkya.gaikwad@students.iiserpune.ac.in}
\equalcont{These authors contributed equally to this work.}

\author[1]{\fnm{Soumen} \sur{Maity}}\email{soumen@iiserpune.ac.in}
\equalcont{These authors contributed equally to this work.}

\author*[1]{\fnm{Leeja} \sur{R}}\email{leeja.r@students.iiserpune.ac.in}
\equalcont{These authors contributed equally to this work.}

\affil[1]{\orgdiv{Department of Mathematics}, \orgname{Indian Institute of Science Education and Research}, \orgaddress{\street{Dr. Homi Bhabha Road}, \city{Pune}, \postcode{411008}, \state{Maharashtra}, \country{India}}}


\abstract{
We study the parameterized complexity of the $\mathcal{T}_{h+1}$-\textnormal{\textsc{Free Edge Deletion}} problem.
Given a graph $G=(V,E)$ and integers $k$ and $h$, the task is to delete at most $k$ edges
so that every connected component of the resulting graph has size at most $h$.
The problem is NP-complete for every fixed $h \geq 3$, while it is solvable in polynomial
time for $h \leq 2$, making it a natural problem to study from a parameterized complexity perspective.

Recent work has revealed significant hardness barriers.
In particular, Bazgan, Nichterlein, and Alferez (JCSS~2025) showed that the problem is W[1]-hard when parameterized by the solution size together with the size of a
feedback edge set, ruling out fixed-parameter tractability for several classical
structural parameters.

We further strengthen the negative results by proving that the problem is W[1]-hard when parameterized by the vertex deletion distance to a disjoint union of paths, the vertex deletion distance to a disjoint union of stars, or the twin cover number of the input graph.
These results extend and unify earlier hardness results for treewidth, pathwidth, and feedback vertex set, and demonstrate that several restrictive structural parameters, including treedepth, cluster vertex deletion number, and modular width, are insufficient to obtain fixed-parameter tractability when $h$ is unbounded.

On the positive side, we show that the problem is fixed-parameter tractable under several
structural parameterizations, including the cluster vertex deletion number together with $h$,
neighborhood diversity together with $h$, and the vertex deletion distance to a clique.
Since the problem is W[1]-hard when parameterized by the solution size $k$ alone,
we also investigate approximation.
We present a fixed-parameter tractable bicriteria approximation algorithm parameterized
by $k$ that either correctly reports that no solution of size at most $k$ exists,
or returns a feasible solution of size at most $4k^{2}$. We also show hardness of the directed generalization of this problem on directed acyclic graphs.
Finally, we study the problem on restricted graph classes.
We show that the problem admits fixed-parameter tractable
algorithms parameterized by $k$ on split graphs and on interval graphs, both of which
form subclasses of chordal graphs.}




\maketitle

\section{Introduction}\label{sec1}

Edge modification problems form a central class of graph problems in algorithmic graph theory, asking whether a small number of edge additions or deletions can transform a given graph into one satisfying a prescribed structural property. Such problems arise naturally in applications ranging from network design~\cite{drange2016vertex,1e396962aa954ff5bbccabe84d44e71c} to computational biology~\cite{10.1007/978-3-642-20877-5_30}, and they have been extensively studied from both classical and parameterized complexity perspectives~\cite{CAI1996171,10.1007/978-3-540-77120-3_79,gaikwad2025parameterizedcomplexitysclubcluster,CaoMarx2016ChordalEditing,MATHIESON20103181,2013arXiv1306.3181C,misra_et_al:LIPIcs.ISAAC.2023.53,gaikwad2025parameterizedalgorithmseditinguniform}. For broader survey on edge modification problems and their algorithmic and complexity-theoretic aspects,
we refer the reader to~\cite{CRESPELLE2023100556}.

In this work, we study the $\mathcal{T}_{h+1}$-\textnormal{\textsc{Free Edge Deletion}} problem. Given an undirected graph $G=(V,E)$ and integers $k$ and $h$, the task is to determine whether at most $k$ edges can be deleted so that every connected component of the resulting graph has size at most $h$. Equivalently, the resulting graph must exclude every tree on $h+1$ vertices as a subgraph. This problem is motivated, for example, by applications in network epidemiology~\cite{bib1}, where bounding the size of connected components limits the worst-case spread of an epidemic through a contact network. Variants of the edge-modification problem considered here have appeared in the literature under different terminology, most notably as the \emph{component order edge connectivity} problem~\cite{Gross2013ComponentOrder} and the \emph{minimum worst contamination} problem~\cite{10.1007/978-3-642-20877-5_30}.

Despite its natural formulation, $\mathcal{T}_{h+1}$-\textnormal{\textsc{Free Edge Deletion}} is computationally hard. The problem is NP-complete for every fixed $h \ge 3$~\cite{bib1}. On the other hand, for $h=1$ and $h=2$ the problem can be solved in polynomial time. In the case $h=1$, every connected component must consist of a single vertex, and hence all edges of the input graph must be deleted. For $h=2$, the problem is equivalent to selecting a maximum matching, since the connected components of the resulting graph must have size at most two. A maximum matching can be computed in polynomial time. Consequently, there is a clear boundary between the values of $h$ for which the problem is solvable in polynomial time and those for which it becomes NP-complete.

From a parameterized viewpoint, the problem admits a kernel with $2kh$ vertices and $2kh^2 + k$ edges, and it is fixed-parameter tractable when parameterized by the vertex cover number of the input graph~\cite{bib2}. Furthermore, Enright and Meeks~\cite{bib1} showed that the problem can be solved in time $\mathcal{O}((h \cdot {tw})^{2tw} \cdot n)$ on graphs of treewidth at most $tw$ and $n$ vertices. However, with the exception of parameterizations based on vertex cover, all known positive results crucially rely on the parameter $h$ being bounded.
This dependence on $h$ raises the following fundamental question:
\emph{Are there structural parameterizations that admit fixed-parameter tractable algorithms for
$\mathcal{T}_{h+1}$-\textnormal{\textsc{Free Edge Deletion}} even when $h$ is unbounded?}

Enright and Meeks explicitly conjectured that $\mathcal{T}_{h+1}$-\textnormal{\textsc{Free Edge Deletion}} is W[1]-hard when parameterized by the treewidth of the input graph alone. This conjecture was later resolved by Gaikwad and Maity~\cite{bib2}, who proved that the problem is indeed W[1]-hard under this parameterization. Subsequently, Bazgan, Nichterlein, and Alferez~\cite{bib11} significantly strengthened this hardness result by showing that the problem remains W[1]-hard when parameterized by the solution size together with the size of a feedback edge set. As a consequence, $\mathcal{T}_{h+1}$-\textnormal{\textsc{Free Edge Deletion}} is W[1]-hard for several classical structural parameters, including feedback vertex set, pathwidth, and treewidth.

\subsection{Our Results}

In this paper, we advance the understanding of $\mathcal{T}_{h+1}$-\textnormal{\textsc{Free Edge Deletion}} for unbounded $h$ in two complementary directions.

First, we show that $\mathcal{T}_{h+1}$-\textnormal{\textsc{Free Edge Deletion}} is W[1]-hard when parameterized by the treedepth or the twin cover of the input graph. This result strictly strengthens the known hardness for pathwidth and shows that even very restrictive structural decompositions do not suffice to obtain fixed-parameter tractability when $h$ is unbounded. This also rules out FPT algorithm when parameterized by other parameters like cluster vertex deletion set and modular width as well. Refer to figure \ref{overview} for overview.

Second, we present positive results based on structural deletion parameters.
Although the problem is W[1]-hard when parameterized by the cluster vertex deletion number alone,
we show that it becomes fixed-parameter tractable when parameterized by the cluster vertex deletion
number $\ell$ together with $h$.
Assuming that a cluster vertex deletion set of size $\ell$ is given, we reduce the problem to
instances of bounded pathwidth and obtain an FPT algorithm whose running time depends only on
$\ell$ and $h$.
Our approach relies on a structural analysis of the components outside the deletion set and on a
reduction rule that bounds their size while preserving optimal solutions.
We further show that the problem is fixed-parameter tractable when parameterized by the vertex
deletion distance to a clique.
Finally, we consider the parameterization by neighborhood diversity.
While the parameterized complexity of
$\mathcal{T}_{h+1}$-\textnormal{\textsc{Free Edge Deletion}} with respect to neighborhood diversity
alone remains open, we show that the problem is fixed-parameter tractable when parameterized by
neighborhood diversity together with $h$. Our algorithm is obtained by reducing the problem to an integer linear program whose size depends only on these parameters.

Finally, we consider the parameterization by the solution size. It is known that
$\mathcal{T}_{h+1}$-\textnormal{\textsc{Free Edge Deletion}} is W[1]-hard when parameterized by the solution size alone. In light of this hardness, we turn to approximation and present a fixed-parameter tractable bicriteria approximation algorithm parameterized by the solution size $k$. Given an instance $(G,k,h)$, if $(G,k,h)$ is a yes-instance, then the algorithm returns
a feasible solution of size at most $4k^{2}$.
If $(G,k,h)$ is a no-instance, the algorithm may either correctly conclude that the instance
is a no-instance or return a feasible solution of size at most $4k^{2}$. The running time of the algorithm is fixed-parameter tractable with respect to $k$. 
We also consider a natural directed generalization of the problem.
As already noted by Enright and Meeks, many motivating applications—most notably
epidemiology and animal movement networks—are inherently directional, and bounding
undirected connected components does not adequately capture worst-case spread.
In the directed setting, a natural analogue of \textsc{$\mathcal{T}_{h+1}$-Free Edge Deletion}
asks whether at most $k$ arcs can be deleted so that, from every vertex, the number of
vertices reachable by directed paths is at most $h$.
We show that this directed variant is W[2]-hard when parameterized by $k$,
even when the input graph is a directed acyclic graph.

Finally, we identify restricted graph classes that are subclasses of chordal graphs, namely split graphs and interval graphs, and show that on these graph classes $\mathcal{T}_{h+1}$-\textnormal{\textsc{Free Edge Deletion}} admits fixed-parameter tractable algorithms when parameterized by the solution size $k$.

\begin{figure}[ht]
\centering
\scalebox{0.7}{
\begin{tikzpicture}

\begin{scope}[xshift=7.5cm]

\node[rectangle, draw, green!60!black, thick] (A2) at (0,-2) {vc};
\node[rectangle, draw, black] (B2) at (1.2,1.5) {nd};
\node[rectangle, draw, red] (C2) at (3.2,1.2){tc};
\node[rectangle, draw, red] (D2) at (1.5,3) {mw};
\node[rectangle, draw, red] (E2) at (0,4) {cw};
\node[rectangle, draw, red] (F2) at (-0.9,2) {pw};
\node[rectangle, draw, red] (G2) at (-1.8,1.5) {fvs};
\node[rectangle, draw, red] (H2) at (-1.2,3) {tw};
\node[rectangle, draw, red] (I2) at (3.2,3) {cvd};
\node[rectangle, draw, red] (J2) at (-0.5,1) {td};

\node[rectangle, draw, red] (K2) at (-2.2,0.2) {vdp}; 
\node[rectangle, draw, red] (L2) at (-0.7,-0.5) {vds}; 
\node[rectangle, draw, green!60!black, thick] (M2) at (0.7,-0.2) {vdc}; 

\draw[->] (A2) -- (C2);
\draw[->] (B2) -- (D2);
\draw[->] (C2) -- (D2);
\draw[->] (D2) -- (E2);
\draw[->] (C2) -- (I2);
\draw[->] (G2) -- (H2);
\draw[->] (F2) -- (H2);
\draw[->] (H2) -- (E2);
\draw[->] (I2) -- (E2);
\draw[->] (J2) -- (F2);

\draw[->] (A2) -- (K2); 
\draw[->] (A2) -- (M2); 
\draw[->] (M2) -- (B2); 
\draw[->] (M2) -- (I2); 
\draw[->] (K2) -- (G2); 
\draw[->] (L2) -- (J2); 
\draw[->] (L2) -- (G2); 

\end{scope}

\end{tikzpicture}
}
\caption{\footnotesize
Relationship between vertex cover [vc] (see Definition \ref{defvc}), neighborhood diversity [nd] (see Definition \ref{defnd}), twin cover [tc] (see Definition \ref{deftc}), modular width [mw] (see~\cite{defmodwidth}), cluster vertex deletion number [cvd] (see Definition \ref{cvd def}), feedback vertex set [fvs] (see Definition \ref{deffvs}), pathwidth [pw] (see Definition \ref{defpw}), treewidth [tw] (see Definition \ref{deftw}) and clique width [cw] (see~\cite{bib14}).  
We additionally include vertex deletion distance to a disjoint union of paths [vdp], vertex deletion distance to a clique[vdc] and
vertex deletion distance to a disjoint union of stars [vds].
Note that $A\rightarrow B$ means that there exists a function $f$ such that for 
all graphs,~$f(A(G))\geq B(G)$. This gives complexity landscape of $\mathcal{T}_{h+1}$-\textnormal{\textsc{Free Edge Deletion}}  under different parameterizations.
Red indicates W[1]-hardness, green indicates fixed-parameter tractability,
and black denotes parameterizations for which the complexity status remains open.}
\label{overview}
\end{figure}
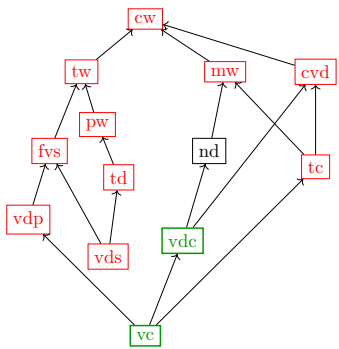

\section{Notation and Definitions}
Unless otherwise stated, all graphs are simple, undirected, and loopless. For a graph $G$, $V(G)$ is the vertex set of $G$ and $E(G)$ the edge set of $G$. We denote the sizes of the vertex and edge sets as $n=|V(G)|$ and $m=|E(G)|$. For vertices $u, v \in V(G)$, we denote the undirected edge between them as $uv$. For a vertex $v \in V(G)$ and a subset $A \subseteq V(G)$, $N_A(v)= \{u \in A : uv \in E(G)\}$ denote the set of neighbors of $v$ that are in $A$. Similarly, for subsets $A,B \subseteq V(G)$, $N_A(B)= \{u \in A : uv \in E(G) \text{ for some } v \in B\}$. We further define
$E(A,B) = \{uv \in E(G) : u \in A,\ v \in B\}$,
that is, the set of edges with one endpoint in $A$ and the other in $B$. Let $E' \subseteq E(G)$ be a set of edges in $G$. Then $G \setminus E'$ denotes the subgraph of $G$ obtained by deleting $E'$ from $G$. We follow standard graph-theoretic notation as in~\cite{13638}. 
We refer to~\cite{marekcygan,Downey} for details on parameterized complexity.

\begin{defn}\label{defvc}
A set $S \subseteq V(G)$ is a \emph{vertex cover} of $G$ if  every edge in $E(G)$ has at least one  endpoint in $S$. The \emph{size} of a smallest vertex cover of  $G$ is the \emph{vertex cover number} of $G$.
\end{defn}

\begin{defn}\label{deffvs}
 A \emph{feedback vertex set}  of a graph  $G$ is a set of vertices whose removal turns $G$ into a forest. The minimum size of a feedback vertex set in $G$ is the \emph{ feedback vertex set number} of $G$.
\end{defn}

\noindent A rooted forest is a disjoint union of rooted trees. Given a rooted forest $Y$, its \emph{closure} is a  graph $H$ where $V(H)=V(Y)$, and $E(H)$ contains an edge between two distinct vertices if and only if one is an ancestor of the other in $Y$.

\begin{defn}\label{deftd}~\cite{bib12}
         The \emph{ treedepth} of a graph $G$ is the minimum height of a rooted forest $Y$ whose closure contains the graph $G$ as a subgraph. It is denoted by $td(G)$.
\end{defn}

\noindent Ganian~\cite{bib13} introduced
a new parameter called twin-cover and showed that it is capable of solving
a wide range of hard problems.

\begin{defn}~\cite{bib13}
An edge $uv\in E(G)$ is a \emph{twin edge} of $G$ if $N[u]=N[v]$. 
\end{defn}

\begin{defn}~\cite{bib13}\label{deftc}
  A set $X\subseteq V(G)$ is a \emph{twin-cover} of $G$ if every edge in~$G$ is either twin edge or incident to a vertex in $X$. The \emph{twin-cover number} of $G$, denoted as $tc(G)$, is the minimum possible size
of a twin-cover of $G$.
\end{defn}

Two distinct vertices $u, v$ are called \emph{true twins} if $N[u]=N[v]$ and \emph{false twins} if $N(u)=N(v)$. We say that $u$ and $v$ have the same \emph{neighborhood type} if they are either true or false twins; such vertices are simply called \emph{twins}.

\begin{defn}~\cite{Lampis}\label{defnd}
A graph  $G$
 has \emph{neighborhood diversity} at most $d$, if there exists a partition of $V(G)$
 into at most $d$
 sets (we call these sets {\it type classes}) such that all the vertices in each set have the same neighborhood type.
\end{defn}

We now review the concept of a tree decomposition introduced by Robertson and Seymour~\cite{Neil}.
Treewidth is a  measure of how ``tree-like'' the graph is.
\begin{defn}[Robertson and Seymour~\cite{Neil}]  A {\it tree decomposition} of a 
graph~$G=(V,E)$  is a tree $T$ together with a 
collection of subsets $X_t$ (called \emph{bags}) of $V$ labeled by the nodes $t$ of $T$ such that 
$\bigcup_{t\in T}X_t=V $ and (1) and (2) below hold:
\begin{enumerate}
			\item For every edge $uv \in E(G)$, there  is some $t$ such that $\{u,v\}\subseteq X_t$.
			\item  (Interpolation Property) If $t$ is a node on the unique path in $T$ from $t_1$ to $t_2$, then 
			$X_{t_1}\cap X_{t_2}\subseteq X_t$.
		\end{enumerate}
	\end{defn}

\begin{defn}\label{deftw}~\cite{Neil} The {\it width} of a tree decomposition is
the maximum value of $|X_t|-1 $ taken over all the nodes $t$ of the tree $T$ of the decomposition.
The \emph{treewidth} $tw(G)$ of a graph $G$  is the  minimum width among all possible tree decompositions of $G$.
\end{defn} 

\begin{defn}\label{defpw} 
    If the tree $T$ of a tree decomposition is a path, then we say that the tree decomposition 
    is a {\it path decomposition}. The \emph{pathwidth}  $pw(G)$ of a graph $G$  is the  minimum width among all possible path decompositions of $G$.
\end{defn}

\begin{defn}\label{cvd def}
   The \emph{cluster vertex deletion number} of a graph is the minimum number of its vertices whose deletion results in a disjoint union of complete graphs. 
\end{defn}

\section{Hardness Results}\label{sec4}
In this section, we establish the parameterized hardness of the problem. We present a parameterized reduction from the \textsc{Unary Bin Packing} problem.
    \begin{mdframed}[linewidth=1pt, roundcorner=5pt, frametitle={\textnormal{\textsc{Unary Bin Packing}}}, nobreak=true]
    \noindent\textbf{Input:} A set of $N$ items along with their integer sizes $a_1, a_2,\dots, a_N$ given in unary, and two positive integers $C$ and $t$.\\
    \noindent \textbf{ Question:} Is it possible to partition the $N$ items into $t$ disjoint bins of capacity $C$ such that the sum of the sizes of the items in each bin does not exceed the capacity?
\end{mdframed}

\begin{thm}\label{thm:distF}
Let $\mathcal{F}$ be any graph class that contains a connected graph on $s$ vertices
for every $s\in\mathbb{N}$.
Then \textsc{$\mathcal{T}_{h+1}$-Free Edge Deletion} is \emph{W[1]-hard} when parameterized by
the vertex deletion distance of the input graph to $\mathcal{F}$.
\end{thm}

\begin{proof}
We give a parameterized reduction from \textsc{Unary Bin Packing},
which is known to be W[1]-hard when parameterized by the number of bins~\cite{bib5}.
Let $I$ be an instance of \textsc{Unary Bin Packing} with $N$ items of sizes
$a_1,\dots,a_N$ (given in unary), $t$ bins, and bin capacity $C$.
Let $a := \sum_{i=1}^{N} a_i$.
We construct an instance $I'=(G,k,h)$ of
\textsc{$\mathcal{T}_{h+1}$-Free Edge Deletion} as follows
(see Figure~\ref{fig:tc}).

\begin{enumerate}[label=(\arabic*)]
    \item Create a vertex set $X=\{v_1,v_2,\dots,v_t\}$.

    \item For each item $i\in[N]$, create a connected graph $A_i\in\mathcal{F}$
          on exactly $a_i$ vertices.
          Let $A := \bigcup_{i=1}^{N} A_i$.

    \item Add all edges between $A$ and $X$.

    \item Set $k := a(t-1)$ and $h := 10C + k$.

    \item For each $j\in[t]$, create a connected graph $H_j\in\mathcal{F}$ on $ h' := h - C - 1$ vertices, and add all edges between $H_j$ and $v_j$.
\end{enumerate}

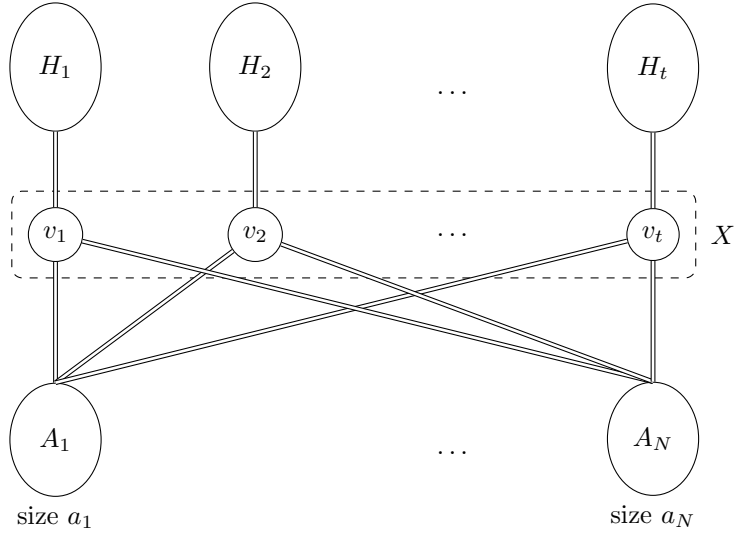
\begin{figure}[H]
    \centering

   \begin{tikzpicture}[
    node distance=1.6cm and 1.9cm,
    v_node/.style={draw, circle, minimum width=0.5cm},
    h_node/.style={draw, ellipse, minimum width=1.2cm, minimum height=1.7cm},
    a_node/.style={draw, ellipse, minimum width=1.2cm, minimum height=1.5cm},
    bundle_edge/.style={double, double distance=1pt, -}
]

\node[v_node] (v1) {$v_1$};
\node[v_node] (v2) [right=of v1] {$v_2$};
\node (v_dots) [right=of v2, draw=none] {$\dots$};
\node[v_node] (vt) [right=of v_dots] {$v_t$};

\node[draw, dashed, rounded corners, fit=(v1) (vt), inner sep=6pt] (X_box) {}; 
\node at (X_box.east) [ xshift=10pt] {\textbf{$X$}};

\node[h_node] (H1) [above= 1cm of v1] {$H_1$};
\node[h_node] (H2) [above=1cm of v2] {$H_2$};
\node (h_dots) [above=of v_dots, draw=none] {$\dots$};
\node[h_node] (Ht) [above=1cm of vt] {$H_t$};

\draw[bundle_edge] (H1) -- (v1);
\draw[bundle_edge] (H2) -- (v2);
\draw[bundle_edge] (Ht) -- (vt);

\node[a_node, label=below:{size $a_1$}] (a1) [below=of v1] {$A_1$};

\node (a_dots) [below=of v_dots, draw=none, yshift=-1cm] {$\dots$};

\node[a_node, label=below:{size $a_N$}] (an) [below=of vt] {$A_N$};

\draw[bundle_edge] (a1.north) -- (v1);
\draw[bundle_edge] (a1.north) -- (v2);
\draw[bundle_edge] (a1.north) -- (vt);

\draw[bundle_edge] (an.north) -- (v1);
\draw[bundle_edge] (an.north) -- (v2);
\draw[bundle_edge] (an.north) -- (vt);

\end{tikzpicture}
   
    \caption{The graph $G$ constructed from the instance $I$ of \textsc{Unary Bin Packing} in Theorem \ref{thm:distF}. A double edge between a vertex $v$ and a set $S$ denotes that $v$ is adjacent to every vertex in $S$.}
    \label{fig:tc}
\end{figure}

There are no edges between $A_i$ and $A_{i'}$ for $i\neq i'$,
no edges between $H_j$ and $H_{j'}$ for $j\neq j'$,
and no edges between any $A_i$ and any $H_j$.
Hence, $G\setminus X$ is a disjoint union of the connected graphs
$A_1,\dots,A_N,H_1,\dots,H_t$.
The construction is computable in polynomial time.

We now prove correctness.
Suppose $I$ is a yes-instance of \textsc{Unary Bin Packing}.
Then there exists a mapping $\gamma:[N]\to[t]$ such that
$\sum_{i\in\gamma^{-1}(j)} a_i \le C$ for all $j\in[t]$.
Construct an edge set $F\subseteq E(G)$ by deleting, for each $i\in[N]$,
all edges between $A_i$ and $X\setminus\{v_{\gamma(i)}\}$.
Clearly, $|F| = a(t-1)=k$.
Let $G' := G\setminus F$.
Each connected component of $G'$ contains exactly one vertex of $X$.
Let $C_j$ denote the component containing $v_j$.
Then
\[
|C_j|
= 1 + |H_j| + \sum_{i\in\gamma^{-1}(j)} |A_i|
= 1 + h' + \sum_{i\in\gamma^{-1}(j)} a_i
\le 1 + (h-C-1) + C
= h.
\]
Thus, $F$ is a feasible solution for $I'$.\\

Suppose $F\subseteq E(G)$ is a solution for $I'$ with $|F|\le k$.
We first claim that no connected component of $G\setminus F$ contains
two vertices of $X$.
Indeed, if a component contains $v_i,v_j\in X$ with $i\neq j$, then
\[
|C| \ge 2 + |H_i| + |H_j| - k
   = 2 + 2h' - k
   = 2(h-C) - k
   = h + 8C > h,
\]
a contradiction.
Hence, every vertex of $A$ is adjacent in $G\setminus F$ to at most one vertex of $X$.
This implies that for every $u\in A$, at least $t-1$ edges incident to $u$ and $X$
must be deleted.
Thus, $|F|\ge a(t-1)=k$, and equality must hold.
Consequently, $F$ consists exactly of these deletions.
Therefore, for each $i\in[N]$ there is a unique $j\in[t]$ such that
$A_i$ is connected to $v_j$ in $G\setminus F$.
Define $\gamma(i)=j$.
Let $C_j$ be the component containing $v_j$.
Then
\[
|C_j|
= |H_j| + 1 + \sum_{i\in\gamma^{-1}(j)} |A_i|
= h' + 1 + \sum_{i\in\gamma^{-1}(j)} a_i
\le h,
\]
which implies $\sum_{i\in\gamma^{-1}(j)} a_i \le C$.
Thus, $\gamma$ defines a valid packing for $I$.
Hence, $I$ is a yes-instance if and only if $I'$ is a yes-instance.
\end{proof}

\begin{cor}\label{cor:masterF}
\textsc{$\mathcal{T}_{h+1}$-Free Edge Deletion} is W[1]-hard when parameterized by each of the following parameters:
\begin{enumerate}[label=(\roman*)]
    \item the size of a vertex set whose deletion turns the graph into a disjoint union of cliques;
    \item the size of a vertex set whose deletion turns the graph into a disjoint union of paths;
    \item the size of a vertex set whose deletion turns the graph into a disjoint union of stars.
\end{enumerate}
Moreover, by choosing all gadgets $A_i$ and $H_j$ to be cliques in the reduction, we obtain that
\textsc{$\mathcal{T}_{h+1}$-Free Edge Deletion} is W[1]-hard parameterized by the twin cover number.
\end{cor}

\begin{cor}
   $\mathcal{T}_{h+1}$-\textnormal{\textsc{Free Edge Deletion}}  is W[1]-hard parameterized by the treedepth, cluster vertex deletion number or modular width of the input graph. 
\end{cor}

\section{FPT Parameterized by Cluster Vertex Deletion and \texorpdfstring{$h$}{h}}\label{sec2}
We assume that a cluster vertex deletion (CVD) set $X$ of size $\ell$ is given with the input graph $G$. 
Otherwise, we can run the algorithm by Boral, Cygan, Kociumaka, and Pilipczuk~\cite{bib9} that either outputs a CVD of size at most $\ell$ or says that no such set exists in $\mathcal{O}(1.9102^{\ell}(n+m))$ time.

Let $F$ be an optimal solution to $\mathcal{T}_{h+1}$-\textnormal{\textsc{Free Edge Deletion}}  on $G$.
For any connected component $C$ of $G \setminus X$, let $C' \subseteq V(C)$ be the set of vertices
that belong to connected components of $G \setminus F$ which do \emph{not} intersect $X$.
The following lemma characterizes the structure induced by the vertices of $C'$ after deleting the optimal edge set $F$.

\begin{lem}\label{lem1}
     If $|C'| = qh+r$ where $q,r \in \mathbb{N }$ and $0 \leq r <h$, then in the graph $G \setminus F$, the set $C'$ induces a disjoint union of $q$ cliques of size $h$ and one clique of size $r$.
\end{lem}

\begin{proof} Let $G'$ be the subgraph of $G \setminus F$ induced by $C'$.
Every connected component of $G'$ has size at most $h$. 
Assume, towards a contradiction, that the lemma does not hold.
Then there exist two components, say $H_1, H_2$, of $G'$ such that  $|V(H_1)| \leq |V(H_2)|  \leq h-1$. 
Define a set $F' \subseteq E(G)$ by reassigning a vertex $v$ of $ H_1$ to the component $H_2$; that is,
$F' = (F \setminus E(v, H_2)) \cup E(v, H_1)$,
where $E(v, H_i)$ denotes set of all edges between $v$ and $H_i$ in $G$.
Then, the connected components of $G \setminus F'$ are $H_1 \setminus \{v\}$ and $H_2 \cup \{v\}$, together with all connected components of $G \setminus F$ other than $H_1$ and $H_2$.
Hence, all connected components of $G \setminus F'$ have size at most $h$.
Moreover, $|F'| = |F| - |V(H_2)| + |V(H_1)| - 1 < |F|$.
Thus, $F'$ is a feasible solution, contradicting our assumption that $F$ is an optimal solution.
Hence, at most one connected component of $G'$ can have size less than $h$, which proves the lemma.
\end{proof}

The specific vertices of $C$ that constitute $C'$ depend on the choice of the optimal solution $F$.
To address this, we partition $C$ into equivalence classes
$P_1, P_2, \dots, P_{2^{\ell}}$ according to their neighborhoods in $X$
(note that some classes may be empty); that is, for $u,v \in C$,
we have $u, v \in P_i$ if and only if $N(u) \cap X = N(v) \cap X$.
We observe that any two vertices in the same class $P_i$ are true twins.
Due to this structural symmetry, swapping the roles of any two vertices in $P_i$ in the graph $G \setminus F$
preserves the size of the solution. 
Consequently, it suffices to keep track of the number of vertices chosen from each class $P_i$, rather than their identities.
The following reduction rule is used to bound the sizes of the connected components of $G \setminus X$.

\begin{redn}\label{redn1}
If $|P_i| \geq \ell (h-1) +h$, then remove an arbitrary set of $h$ vertices of $P_i$ from $G$ and decrease the parameter $k$ by $h \cdot (|C|-h +|N_{X}(P_i)|)$. If, after applying the rule, the parameter $k$ becomes negative, then return a no-instance.
\end{redn}
\begin{lem}
 Reduction rule \ref{redn1} is safe.
\end{lem}
\begin{proof}
Let $(G,k,h)$ be an instance of $\mathcal{T}_{h+1}$-\textnormal{\textsc{Free Edge Deletion}} , let $X$ be a cluster vertex deletion set of size $\ell$ of $G$, and $C$ be a connected component of $G \setminus X$. 
Suppose $C$ is  partitioned into $P_1,\dots, P_{2^\ell}$ according to their neighborhood in $X$, and let $|P_i|\geq \ell (h-1) +h$ for some $i$. Let $G'$ be the new graph obtained by removing an arbitrary set $S \subseteq P_i$ of $h$ vertices from $G$ and let $k'=k-h \cdot (|C|-h +|N(P_i)|)$. We show that $(G, k,h)$ is a yes-instance if and only if $(G',k',h)$ is a yes-instance. 

Suppose $(G,k,h)$ is a yes-instance, and let $F$ be an optimal solution. 
Let $C' \subseteq C$ denote the set of vertices that are contained in connected components of $G \setminus F$ which do not intersect $X$.
Observe that at most $\ell (h-1)$ vertices can belong to $C \setminus C'$ because any component of $G \setminus F$ that intersects with $X$ can contain at most $h-1$ vertices of $C$. So if $|P_i| \geq \ell (h-1) +h$, we have $ |C' \cap P_i| \geq h$. 
Since $|C'| \geq h$, Lemma \ref{lem1} implies that the subgraph of $G \setminus F$ induced by $C'$ contains at least one connected component that is a clique of size $h$. 
Note that all edges between $C'$ and $V(G) \setminus C'$ are contained in $F$.
This, together with the fact that $C'$ induces a clique in $G$, implies that the specific vertices forming such a component can be chosen arbitrarily from $C'$. 
Therefore, since $|C' \cap P_i| \ge h$, we may assume without loss of generality that the vertex set $H$ inducing such a component is contained in $C' \cap P_i$.
The number of edges between $H$ and $V(G) \setminus H$ in $G$ is $$|E(H, V(G) \setminus H)|= |H|(|C|-|H|) + |H||N_{X}(P_i)| = h (|C|-h) + h |N_{X}(P_i)|.$$
Define $F'=F \setminus E(H, G\setminus H)$.
Clearly, $|F'| \leq k'$, and all connected components of $(G \setminus H) \setminus F'$ has size at most $h$.
So $(G \setminus H, k', h)$ is a yes-instance. 
Since all vertices of $P_i$ have identical closed neighborhoods in $G$, $G \setminus H$ and $G'$ are isomorphic.
Hence, $(G', k', h)$ is a yes-instance.

Conversely, let $(G', k', h)$ be a yes-instance, and let $F'$ be an optimal solution.
Let $E(S, G')$ denote the set of all edges in $G$ with one endpoint in $S$ and the other in $V(G')$.
Then,
\[
|E(S, G')| = |S|(|C| - |S|) + |S|\cdot |N_{X}(P_i)| = h(|C| - h) + h\cdot |N_{X}(P_i)|.
\]
Now consider the set $F = F' \cup E(S, G \setminus S)$.
Then $|F| \le k$, and the connected components of $G \setminus F$ are exactly the connected components of $G' \setminus F'$ together with $G[S]$, all of which have size at most $h$.
Hence, $(G, k, h)$ is a yes-instance.
\end{proof}
 
 \begin{thm}
    $\mathcal{T}_{h+1}$-\textnormal{\textsc{Free Edge Deletion}}  can be solved in $\mathcal{O}\!\left((2^{\ell}\ell h^{2})^{\mathcal{O}(2^{\ell}\ell h)} \cdot (n+m)\right)$ time.
\end{thm}
\begin{proof}
Given an instance $(G,k,h)$ of $\mathcal{T}_{h+1}$-\textnormal{\textsc{Free Edge Deletion}}  with $n$ vertices and a cluster vertex deletion set $X$ of size $\ell$, we apply Reduction Rule~\ref{redn1} exhaustively.
If, at any stage, the parameter $k$ becomes negative, we return a no-instance.
Otherwise, in the reduced graph $G$, every connected component of $G \setminus X$ has size at most $2^{\ell}(\ell+1)(h-1)$.
Let $C_1, C_2, \dots, C_r$ denote the connected components of $G \setminus X$.
Then $G$ admits a path decomposition with bags $B_i = X \cup C_i$ for all $i \in [r]$.
Hence, the pathwidth of $G$ is at most $\ell + 2^{\ell}(\ell+1)(h-1)$.
We solve the problem using dynamic programming on the constructed path decomposition. The algorithm follows the standard formulation for $\mathcal{T}_{h+1}$-\textnormal{\textsc{Free Edge Deletion}}  on tree decompositions given by~\cite{bib1}, with the simplification that our decomposition contains no join nodes.

The connected components of $G \setminus X$ can be computed using a standard graph traversal,
such as breadth-first search (BFS) or depth-first search (DFS), in time $\mathcal{O}(n+m)$.
The partitioning of these connected components into equivalent classes can be done using a partition refinement algorithm; We initialize the partition $\mathcal{P}$ with the vertex sets of the connected components of $G \setminus X$. We then iteratively refine $\mathcal{P}$ using the neighborhoods of each vertex $x \in X$ as a pivot. By processing each edge incident to $X$ exactly once, the total running time is linear in the size of the graph.
Applying the reduction rule on these partitions can also be done in linear time. Since the states and transition functions for introduce and forget nodes in the dynamic programming algorithm remain identical, the running time bound of $\mathcal{O}((wh)^w \cdot n)$, where $w$ is the pathwidth, follows directly. Hence $\mathcal{T}_{h+1}$-\textnormal{\textsc{Free Edge Deletion}}  can be solved in $\mathcal{O}\!\left((2^{\ell}\ell h^{2})^{\mathcal{O}(2^{\ell}\ell h)} \cdot n+m\right)$ time.
\end{proof}

\begin{thm}\label{thm:vdclique_iqp}
\textsc{$\mathcal{T}_{h+1}$-Free Edge Deletion} is fixed-parameter tractable when
parameterized by the size $\ell$ of a vertex deletion set to a clique.
\end{thm}

\begin{proof}
Let $(G,k,h)$ be an instance of \textsc{$\mathcal{T}_{h+1}$-Free Edge Deletion}.
Let $X\subseteq V(G)$ be a given vertex deletion set of size $\ell$ such that
$C:=V(G)\setminus X$ induces a clique. 
For each subset $S\subseteq X$, define
$P_S := \{v\in C : N_X(v)=S\}$.
Then $\{P_S : S\subseteq X\}$ is a partition of $C$ into at most $2^\ell$ classes,
and any two vertices within the same class are true twins.
We apply Reduction Rule~\ref{redn1} exhaustively to $G$ (with respect to the fixed set $X$).
If at any point the parameter $k$ becomes negative, we return a \textsc{No}-instance.
Otherwise, in the reduced instance we have 
$|P_S| \le \ell(h-1)+h = (\ell+1)h-\ell$ for every $S\subseteq X$.
Consequently,
\begin{equation}\label{eq:Vbound}
|C| \le 2^\ell\big((\ell+1)h-\ell\big) \le 2^\ell(\ell+1)h
\quad\text{and hence}\quad
|V(G)| \le \ell + 2^\ell(\ell+1)h.
\end{equation}

A feasible solution $F$ defines the graph $H:=G\setminus F$, whose connected components
all have size at most $h$.
We will argue that we may assume an optimal solution admits a \emph{partition into parts}
where all but at most one part have size at least $\lceil h/2\rceil$; this yields a bound
on the number of parts depending only on $\ell$.

\begin{claim}\label{claim:parts}
There exists an optimal solution $F^\star$ such that in $H^\star := G\setminus F^\star$,
the vertex set $V(G)$ can be partitioned into parts $Q_1,\dots,Q_p$ satisfying:
\begin{enumerate}
  \item $|Q_r|\le h$ for every $r\in[p]$, and
  \item all but at most one part satisfy $|Q_r|\ge \lceil h/2\rceil$.
\end{enumerate}
\end{claim}
\begin{claimproof}
Let $F$ be an optimal solution and let $H:=G\setminus F$.
Each connected component of $H$ has size at most $h$.
We form parts by repeatedly merging pairs of connected components of size at most
$\lfloor h/2\rfloor$ as long as possible.
Each merge preserves the size bound $h$, and when the process terminates, at most one
part has size at most $\lfloor h/2\rfloor$, proving Item~(2).

Let $Q_1,\dots,Q_p$ denote the resulting parts and define
$F^\star$ by deleting all edges between distinct parts.
Since each part is a union of connected components of $H$, every edge between two
different parts must already belong to $F$, implying $F^\star\subseteq F$ and
$|F^\star|\le |F|$.
Thus $F^\star$ is optimal.
Moreover, every connected component of $G\setminus F^\star$ is contained in a single
part and hence has size at most $h$, proving Item~(1).
\end{claimproof}

\begin{claim}
   In Claim~\ref{claim:parts}, we may assume the number of parts satisfies $$p \le f(\ell) := 2^{\ell+1}(\ell+1)+2\ell+1$$.
\end{claim}
\begin{claimproof}
    By Claim~\ref{claim:parts}, all but at most one part have size at least $\lceil h/2\rceil$.
Thus
\[
(p-1)\cdot \left\lceil \frac{h}{2}\right\rceil \le |V(G)|-1
\quad\Rightarrow\quad
p \le \left\lceil \frac{2|V(G)|}{h}\right\rceil + 1.
\]
Using~\eqref{eq:Vbound},
\[
p \le \left\lceil \frac{2(\ell+2^\ell(\ell+1)h)}{h}\right\rceil + 1
= \left\lceil \frac{2\ell}{h}+2^{\ell+1}(\ell+1)\right\rceil + 1
\le 2^{\ell+1}(\ell+1)+2\ell+1,
\]
where we used $h\ge 1$ to bound $2\ell/h\le 2\ell$.
\end{claimproof}

By Claim~2, there exists an optimal solution that admits a partition into
$p\le f(\ell)$ parts satisfying Claim~1.
We iterate over all integers $p\in\{1,2,\dots,f(\ell)\}$.
For each such $p$, we also iterate over a choice $s\in\{0,1,2,\dots,p\}$, where
$s=0$ means ``there is no small part'' and $s\in[p]$ designates the unique small part.
For each guess $(p,s)$, we construct and solve an Integer Quadratic Program (IQP)
whose optimum equals the minimum number of deleted edges among all feasible partitions
consistent with the guess. We accept if and only if this optimum is at most $k$.

Now, we give Integer Quadratic Programming (IQP) formulation for a fixed guess $(p,s)$. We introduce the following variables.
For each part $r\in[p]$ and each $S\subseteq X$, introduce an integer variable $x_{r,S}\in \mathbb{Z}_{\ge 0}$,
interpreted as the number of vertices from the twin class $P_S$ assigned to part $r$.
For each part $r\in[p]$ and each $u\in X$, introduce a binary variable $y_{r,u}\in\{0,1\}$,
interpreted as $y_{r,u}=1$ if and only if $u$ is assigned to part $r$.
Define shorthand:
\[
c_r := \sum_{S\subseteq X} x_{r,S},\qquad
x_r := \sum_{u\in X} y_{r,u},\qquad
t_r := c_r + x_r.
\]

\smallskip
\noindent\emph{Linear constraints.}
\begin{alignat}{2}
\sum_{r=1}^p x_{r,S} &= |P_S| \qquad &&\forall S\subseteq X, \label{con:class}\\
\sum_{r=1}^p y_{r,u} &= 1 \qquad &&\forall u\in X, \label{con:Xassign}\\
t_r &\le h \qquad &&\forall r\in[p], \label{con:ub}\\
t_r &\ge \left\lceil \frac{h}{2}\right\rceil \qquad &&\forall r\in[p]\text{ with } r\ne s \text{ and } s\ne 0. \label{con:lb}
\end{alignat}
If $s=0$, then we enforce~\eqref{con:lb} for all $r\in[p]$.
All coefficients in these constraints belong to $\{-1,0,1\}$
We minimize the number of edges crossing between different parts. Let $E_X:=E(G[X])$.

\begin{itemize}
\item \textbf{Edges inside the clique $C$.}
Since $C$ is a clique, the number of edges with endpoints in different parts equals
\[
\mathrm{cut}_C := \sum_{1\le r<q\le p} c_r\,c_q.
\]

\item \textbf{Edges between $X$ and $C$.}
A vertex $u\in X$ is adjacent to precisely those vertices of $C$ in classes $P_S$
with $u\in S$. The number of crossing edges between $u$ (placed in part $r$) and clique
vertices placed in a different part $q\ne r$ equals
\[
\sum_{\substack{q=1\\q\ne r}}^p\;\sum_{\substack{S\subseteq X\\ u\in S}} x_{q,S}.
\]
Thus
\[
\mathrm{cut}_{X,C}
:= \sum_{u\in X}\sum_{r=1}^p
y_{r,u}\cdot
\left(\sum_{\substack{q=1\\q\ne r}}^p\;\sum_{\substack{S\subseteq X\\ u\in S}} x_{q,S}\right).
\]

\item \textbf{Edges inside $X$.}
An edge $\{u,v\}\in E_X$ is crossing iff $u$ and $v$ are assigned to different parts. Hence
\[
\mathrm{cut}_X
:= \sum_{\{u,v\}\in E_X}\;\sum_{1\le r<q\le p}
\bigl(y_{r,u}y_{q,v}+y_{q,u}y_{r,v}\bigr).
\]
\end{itemize}

We set the IQP objective to
\[
\min\ \ \mathrm{cut}_C+\mathrm{cut}_{X,C}+\mathrm{cut}_X.
\]
Every coefficient appearing in the quadratic terms is in $\{0,1\}$, and there are no
negative coefficients.

\begin{claim}
   For a fixed guess $(p,s)$, the optimum value of the IQP equals the minimum size of an
edge set $F$ such that $G\setminus F$ has a partition into $p$ parts satisfying
Constraints~\eqref{con:class}--\eqref{con:lb}, and every connected component of
$G\setminus F$ has size at most $h$. 
\end{claim}
\begin{claimproof}
   Consider any feasible assignment to the IQP variables.
Constraints~\eqref{con:class} and~\eqref{con:Xassign} define a partition of $V(G)$ into
$p$ parts by assigning exactly $x_{r,S}$ vertices from each class $P_S$ and exactly one
copy of each $u\in X$ to some part $r$.
Constraints~\eqref{con:ub}--\eqref{con:lb} ensure each part has size at most $h$ and that
all but possibly one part have size at least $\lceil h/2\rceil$.

Let $F$ be the set of edges of $G$ with endpoints in different parts.
Then $G\setminus F$ has no edges between different parts, so every connected component of
$G\setminus F$ is contained within a single part and therefore has size at most $h$.
Moreover, the objective exactly counts $|F|$ by summing: crossing edges inside $C$,
crossing edges between $X$ and $C$, and crossing edges inside $X$.
Thus every feasible IQP solution yields a feasible deletion set whose size equals the
objective value.

Conversely, any partition into $p$ parts satisfying the size constraints induces values
$x_{r,S}$ and $y_{r,u}$ meeting Constraints~\eqref{con:class}--\eqref{con:lb}, and the
objective again equals the number of edges crossing between parts, i.e., the size of the
corresponding deletion set.
Therefore the IQP optimum equals the minimum feasible deletion size under the fixed guess. 
\end{claimproof}

For each guess $(p,s)$, the number of variables is
\[
N_{\mathrm{var}} = p\cdot 2^\ell + p\cdot \ell \le f(\ell)\cdot (2^\ell+\ell),
\]
which depends only on $\ell$ because $p\le f(\ell)$.
Furthermore, the maximum absolute value of any coefficient in the constraints and the
quadratic objective is $1$.
By the result of Lokshtanov~\cite{lokshtanov2017parameterizedintegerquadraticprogramming}, \textsc{Integer Quadratic Programming} is
fixed-parameter tractable when parameterized by $N_{\mathrm{var}}+\alpha$, where
$\alpha$ is the largest absolute value of a coefficient.
Thus each IQP instance can be solved in $g(\ell)\cdot |V(G)|^{O(1)}$ time for some
computable function $g$.
Since we try at most
\[
\sum_{p=1}^{f(\ell)} (p+1) \le (f(\ell)+1)^2
\]
guesses $(p,s)$, the overall running time is still of the form
$g'(\ell)\cdot |V(G)|^{O(1)}$. We accept if $\min_{p,s} \mathrm{OPT}_{p,s}\le k$ and reject otherwise.
Therefore, \textsc{$\mathcal{T}_{h+1}$-Free Edge Deletion} is fixed-parameter tractable
parameterized by $\ell$.
\end{proof}

\section{FPT Parameterized by Neighborhood Diversity and \texorpdfstring{$h$}{h}}\label{sec3}
We formulate the problem as an Integer Linear Programming problem and use an algorithm that solves parameterized minimization ILP in FPT time parameterized by number of variables. 
\begin{framed}
 \noindent   $p$-\textsc{Variable Integer Linear Programming Optimization} ($p$-\textsc{Opt}-ILP)\\
   \textbf{Input:}
matrices $A \in \mathbb{Z}^{m\times p}, b \in \mathbb{Z}^{
m \times 1}$ and $c \in \mathbb{Z}^{
1\times p}$.\\
\textbf{Output:} A vector $\textbf{x} \in \mathbb{Z}^{ 
p\times 1}$
that minimizes the objective function $c\cdot \textbf{x}$ and satisfies the $m$ inequalities,
that is, $A \cdot \textbf{x} \geq b$.
\end{framed}
\begin{lem}\label{lem4}~\cite{bib4}
    p-\textsc{Opt}-ILP can be solved using $\mathcal{O}(p^{2.5p+o(p)}
\cdot L \cdot \log (MN))$ arithmetic operations and space polynomial in $L$. Here, $L$ is the number of bits in the input, $N$ is the
maximum of the absolute values any variable can take, and $M$ is an upper bound on the
absolute value of the minimum taken by the objective function.
\end{lem}
\begin{thm}
    $\mathcal{T}_{h+1}$-\textnormal{\textsc{Free Edge Deletion}} can be solved in  time  \[\mathcal{O}\left( f(t,h) \cdot (\log n)^2 + n + m \right)\] where $f(t,h) =t \cdot \binom{h+t}{t}^{2.5 \binom{h+t}{t} + o\left(\binom{h+t}{t}\right)}$.
\end{thm}
\begin{proof}

Let $t$ be the neighborhood diversity of $G$ and $P=(P_1,\dots, P_t)$ be the partition of $V(G)$ into neighborhood types. By definition, each $P_i$ is either a clique or an independent set. Furthermore, for any $i \neq j$, either all vertices of $P_i$ are adjacent to all vertices of $P_j$ or no vertex of $P_i$ is adjacent to any vertex of $P_j$. Without loss of generality, assume that $P_1,\dots,P_r$ are cliques, and $P_{r+1},\dots, P_t$ are independent sets. 

We aim to minimize the size of $F \subseteq E(G)$ such that each component of $G'= G\setminus F$ has at most $h$ vertices. We first characterize the components of $G'$ by its ``type"- a vector that represents its intersection with the neighborhood classes.  Let $$\mathcal{A}= \{\mathbf{a}=(a_1, a_2,\dots, a_t) \in \mathbb{N}_0^t: 0 \leq a_i \leq |P_i|, 1 \leq \sum_{i=1}^t a_i \leq h \}.
$$ A component $C$ of $G'$ is said to be of type $\mathbf{a}$ if $|C \cap P_i|=a_i$ for all $1 \leq i \leq t$. For each $\mathbf{a} \in \mathcal{A}$, define a non-negative integer variable $x_{\mathbf{a}}$ that represents the number of components of type $\mathbf{a}$ in $G'$.
The vertices of each neighborhood type are distributed among the components of $G'$. Thus we have the following constraint: 
$$\sum_{\mathbf{a} \in \mathcal{A}} a_i x_{\mathbf{a}} = |P_i| \text{ for all } i=1,2,\dots,t$$
Now for every $1 \leq i < j \leq t$, define $$z_{ij}= \begin{cases}
     1, \text{ if there is an edge between } P_i \text{ and } P_j \\
     0, \text{ otherwise }
\end{cases}$$
Then, $$|E(G') | = \sum_{\mathbf{a} \in \mathcal{A}} x_{\mathbf{a}} \cdot \Big(\sum_{i=1}^r \binom{a_i}{2} + \sum_{1 \leq i < j \leq t}z_{ij}a_ia_j \Big)$$
Thus given below is the ILP formulation of the $\mathcal{T}_{h+1}$-\textnormal{\textsc{Free Edge Deletion}}  where the partition of $V(G)$ into neighborhood types $P=(P_1,\dots,P_t) $  is given: 
\begin{framed}
 \hspace{1cm}    Minimize $$|E(G)| - \Big[\sum_{a \in \mathcal{A}} x_a \cdot \Big(\sum_{i=1}^r \binom{a_i}{2} + \sum_{1 \leq i < j \leq t}z_{ij}a_ia_j \Big) \Big] $$ 
\hspace{1.5cm} subject to   
        $$ \sum_{a \in \mathcal{A}} a_i x_a = |P_i|   \text{  for all } i=1,2,\dots,t$$
\end{framed}

For a given instance of $\mathcal{T}_{h+1}$-\textnormal{\textsc{Free Edge Deletion}} , we can compute the neighborhood types of the graph in
linear time using fast modular decomposition algorithms~\cite{bib10}. The ILP formulation comprises at most $\binom{h+t}{t}$ variables and $t$ constraints. Note that the values of $z_{ij}$ and $|P_i|$ are calculated in linear time alongside the neighborhood partition. So the construction of the ILP instance takes $\mathcal{O}(t \cdot \binom{h+t}{t})$ time. The value of objective function is bounded by $n^2$ and the value that can be taken by any variable is bounded by $n$. Also, the ILP can be represented using at most $\mathcal{O}(t \cdot \binom{h+t}{t} \cdot \log n )$ bits. Therefore, by Lemma \ref{lem4}, the ILP can be solved in time 
$$\mathcal{O}\left( t \cdot \binom{h+t}{t}^{2.5\binom{h+t}{t}+o(\binom{h+t}{t})} \cdot  (\log n)^2 \right)$$

Hence the $\mathcal{T}_{h+1}$-\textnormal{\textsc{Free Edge Deletion}}  problem can be solved in time $$\mathcal{O}\left( f(t,h) \cdot (\log n)^2 + n + m \right)$$
where $$f(t,h) =t \cdot \binom{h+t}{t}^{2.5 \binom{h+t}{t} + o\left(\binom{h+t}{t}\right)} $$
\end{proof}

\section{An FPT Approximation Algorithm Parameterized by Solution Size}\label{sec8}

Given the W[1]-hardness of the problem parameterized by $k$~\cite{bib11}, it is natural to ask whether the problem admits an efficient approximation.
In this section, we present a parameterized approximation algorithm that returns a solution of size at most $4k^2$. 
We utilize the (FPT) algorithm for \textsc{Minimum Bisection} by Cygan, Lokshtanov, Pilipczuk, Pilipczuk and Saurabh ~\cite{bib8} to iteratively prune "giant" components from the graph. 
They presented an algorithm that, given a connected graph $G$ and an integer $k$, can find a subset $A \subseteq V(G)$ such that $|A| = t$ and $|E(A, V(G) \setminus A)| \leq k$ for every integer $t \in \{1, \dots, n\}$ (or correctly report that no such subset exists).
Note that the algorithm computes these solutions for all possible values of $t$ simultaneously in a single execution taking $\mathcal{O}(2^{\mathcal{O}(k^3)} n^3 \log ^3n)$ time.

\begin{thm}
There exists an algorithm that, given an instance $(G,k,h)$ of
\textsc{$\mathcal{T}_{h+1}$-Free Edge Deletion}, runs in time
$\mathcal{O}(2^{\mathcal{O}(k^3)} n^3 \log^3 n)$ and satisfies the following:
\begin{itemize}
    \item if $(G,k,h)$ is a yes-instance, then the algorithm returns
    a feasible solution $F'$ with $|F'| \leq 4k^2$;
    \item if $(G,k,h)$ is a no-instance, then the algorithm may either
    correctly conclude that no solution of size at most $k$ exists or return
    a feasible solution $F'$ with $|F'| \leq 4k^2$.
\end{itemize}
\end{thm}

\begin{proof}
First, we apply the kernelization algorithm by Gaikwad and Maity~\cite{bib2} which runs in linear time in the input size. This either returns a no-instance or gives a reduced instance with at most $2kh$ vertices. We proceed to run the following procedure on the reduced instance: 
Initialize $F'=\emptyset$. While there exists a connected component of size greater than $h$, run the algorithm by Cygan et al.~\cite{bib8} on $C$. If the algorithm finds a set $A \subseteq V(C)$ with $|A|=t$ for some $h/2 \leq t \leq h$ and $|E(A, V(G) \setminus A)| \leq k$, then add the edges $E(A, V(G) \setminus A)$ to $F'$, remove $A$ from $G$, and continue to the next iteration of the while loop. If the algorithm fails to find such a set for all $t\in [h/2,h]$, return no-instance. This is because if $(G, k, h)$ is a yes-instance, any component $C$ of size greater than $h$ must contain a subset $A$ of size at most $ h$ that can be separated by at most $k$ edges. If $|A| \geq h/2$, the algorithm will return $A$, and if all such $A$ has $|A| \leq h/2$, then we can merge some of them too get the size greater than $h/2$. 

In every successful iteration, we identify and remove a vertex set $A$ of size at least $h/2$. Since we start with the kernelized graph where $|V(G)| \leq 2kh$, the maximum number of iterations $R$ is bounded by:
$$R \le \frac{|V(G)|}{h/2} \leq\frac{2kh}{h/2} = 4k$$
If all iterations were successful, we return $F'$. Since we add at most $k$ edges to $F'$ in each iteration, the total size of the returned solution is bounded by:
$$|F'| \le R \cdot k \le 4k^2$$

 The algorithm by Cygan et al. runs in $\mathcal{O}(2^{\mathcal{O}(k^3)} n^3 \log ^3n)$ time. Since the loop calls for this algorithm at most $4k$ times, the total running time is $\mathcal{O}(2^{\mathcal{O}(k^3)} k n^3 \log ^3n)$.
    
\end{proof}

\noindent The natural generalization of this problem to directed graphs in this context would be to consider whether it is possible to delete at most $k$ arcs from a given directed graph so that the maximum number of vertices reachable from any given starting vertex is at most $h$.
 \vspace{3mm}
    \\
    \fbox
    {\begin{minipage}{33.7em}\label{BSP2}
       {\sc $\mathcal{T}_{h+1}$-Free  Arc Deletion}\\
        \noindent{\bf Input:} A directed graph $G=(V,E)$, and two positive integers $k$ and $h$.\\
    \noindent{\bf Question:} Does there exist $E'\subseteq E(G)$ with 
    $|E^{\prime}|\leq k$ such that the maximum number of vertices reachable from any given starting vertex is at most $h$ in     
    $G\setminus E^{\prime}$?  
    \end{minipage} }\\
 \vspace{3mm}   

Enright and Meeks  \cite{bib1} explained the importance of studying {\sc $\mathcal{T}_{h+1}$-Free Arc Deletion}. 
A directed acyclic graph (DAG) is a directed graph with no directed cycles. One natural problem mentioned in \cite{bib1} is to consider whether there exists an FPT algorithm to solve this problem on directed acyclic graphs. We show that the problem is W[2]-hard  parameterized by the solution size $k$, even when restricted to directed acyclic graphs (DAG). We prove this result via a reduction from 
{\sc Hitting Set}. In the {\sc Hitting Set} problem we are given as input a family $\mathcal{F}$ over a universe $U$, together with an integer $k$, and the objective is to determine whether there is a set $B \subseteq U$ of size at most $k$ such that $B$ has nonempty intersection with all sets in $\mathcal{F}$. It is proved in \cite{marekcygan} (Theorem 13.28) that {\sc Hitting Set} problem is W[2]-hard  parameterized by the solution size.

\begin{thm}\label{theorem-W[2]}
The {\sc $\mathcal{T}_{h+1}$-Free Arc Deletion} problem is W[2]-hard  parameterized by the solution size $k$, even when restricted to directed acyclic graphs.
\end{thm}

\begin{proof}
Let $I=(U,\mathcal{F},k)$ be an instance of {\sc Hitting Set} where $U=\{x_1,x_2,\ldots,x_n\}$. We create an instance $I'=(G',k',h)$ of {\sc $\mathcal{T}_{h+1}$-Free Arc Deletion} the following way.
For every $x\in U$, create two vertices $v_x$ and $v'_x$ and  add a directed edge $(v_x,v'_x)$.
For every $F\in \mathcal{F}$, create one vertex $v_F$. 
Next, we add a directed edge $(v_F,v_x)$ if and only if $x\in F$. 
 For each  $x\in U$,
we add a set $V_x$ of $\frac{h}{n}$ many new vertices and 
 add a directed edge from $v'_x$  to
every vertex of $V_x$.
 We specify the value of $h$ at the end of the construction.
For each vertex $F\in \mathcal{F}$, we add a set $V_{F}$ of 
$(h+1)- \sum\limits_{x\in F} |V_x|$ new vertices and add a directed edge from
$v_F$ to 
every vertex of $V_{F}$. 
 Finally, we set $k'=k$ and  $h=n^{c}$ for some large constant $c$. 
This completes the construction of $G'$. Next, we show that $I$ and $I'$ are equivalent instances.

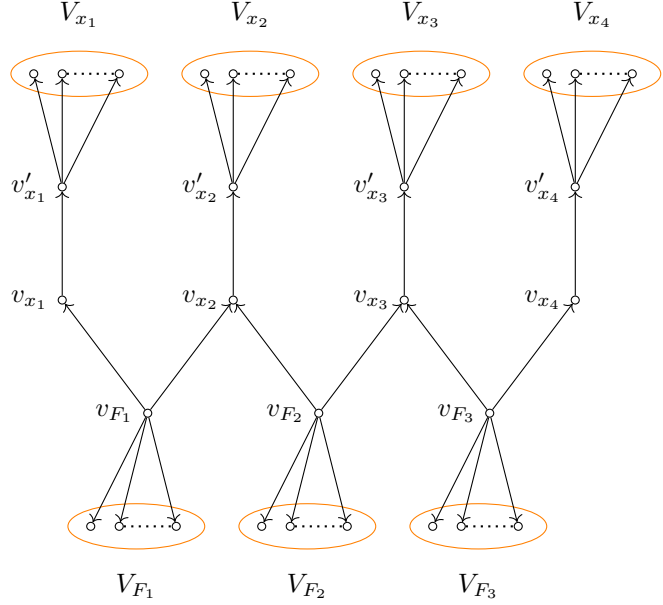
\begin{figure}[t]
\centering
\begin{tikzpicture}[scale=1.5]\centering

\vertex (vx1) at (0, 0) [label=left:$v_{x_{1}}$]{};  

\vertex (vx11) at (-0.25, 2) [label=left:]{}; 
\vertex (vx12) at (0, 2) [label=left:]{}; 
\vertex (vx13) at (.5, 2) [label=left:]{};

\draw[thick, dotted] (vx12)--(vx13);

\draw [orange](0.15,2) ellipse (.6 cm and .2cm);

\node (1) at (0.15, 2.25) [label=above:$V_{x_{1}}$]{};
\node (1) at (-2, 0) [label=above:]{};

\vertex (vx2) at (1.5, 0) [label=left:$v_{x_{2}}$]{};

\vertex (vx21) at (1.25, 2) [label=left:]{}; 
\vertex (vx22) at (1.5, 2) [label=left:]{}; 
\vertex (vx23) at (2, 2) [label=left:]{};

\draw[thick, dotted] (vx22)--(vx23);

\draw [orange](1.65,2) ellipse (.6 cm and .2cm);

\node (2) at (1.65, 2.25) [label=above:$V_{x_{2}}$]{};

\vertex (vx3) at (3, 0) [label=left:$v_{x_{3}}$]{};

\vertex (vx31) at (2.75, 2) [label=left:]{}; 
\vertex (vx32) at (3, 2) [label=left:]{}; 
\vertex (vx33) at (3.5, 2) [label=left:]{};

\draw[thick, dotted] (vx32)--(vx33);

\draw [orange](3.15,2) ellipse (.6 cm and .2cm);

\node (3) at (3.15, 2.25) [label=above:$V_{x_{3}}$]{};

\vertex (vx4) at (4.5, 0) [label=left:$v_{x_{4}}$]{};

\vertex (vx41) at (4.25, 2) [label=left:]{}; 
\vertex (vx42) at (4.5, 2) [label=left:]{}; 
\vertex (vx43) at (5, 2) [label=left:]{};

\draw[thick, dotted] (vx42)--(vx43);

\draw [orange](4.65,2) ellipse (.6 cm and .2cm);

\node (4) at (4.65, 2.25) [label=above:$V_{x_{4}}$]{};

\vertex (vf1) at (0.75, -1) [label=left:$v_{F_{1}}$]{};

\vertex (vf11) at (0.25, -2) [label=left:]{}; 
\vertex (vf12) at (0.5, -2) [label=left:]{}; 
\vertex (vf13) at (1, -2) [label=left:]{}; 

\draw [orange](0.65,-2) ellipse (.6 cm and .2cm);

\node (1) at (0.65, -2.25) [label=below:$V_{F_{1}}$]{};

\draw[thick, dotted] (vf12)--(vf13);

\vertex (vf2) at (2.25, -1) [label=left:$v_{F_{2}}$]{};  

\vertex (vf21) at (1.75, -2) [label=left:]{}; 
\vertex (vf22) at (2, -2) [label=left:]{}; 
\vertex (vf23) at (2.5, -2) [label=left:]{}; 

\node (2) at (2.15, -2.25) [label=below:$V_{F_{2}}$]{};

\draw [orange](2.15,-2) ellipse (.6 cm and .2cm);

\draw[thick, dotted] (vf22)--(vf23);

\vertex (vf3) at (3.75, -1) [label=left:$v_{F_{3}}$]{};  

\vertex (vf31) at (3.25, -2) [label=left:]{}; 
\vertex (vf32) at (3.5, -2) [label=left:]{}; 
\vertex (vf33) at (4, -2) [label=left:]{};

\draw [orange](3.65,-2) ellipse (.6 cm and .2cm);
   
\node (3) at (3.65, -2.25) [label=below:$V_{F_{3}}$]{};

\draw[thick, dotted] (vf32)--(vf33);

\vertex (vx01) at (0, 1) [label=left:$v^{\prime}_{x_{1}}$]{};  
\vertex (vx02) at (1.5, 1) [label=left:$v^{\prime}_{x_{2}}$]{};
\vertex (vx03) at (3, 1) [label=left:$v^{\prime}_{x_{3}}$]{};
\vertex (vx04) at (4.5, 1) [label=left:$v^{\prime}_{x_{4}}$]{};

\draw[->] (vx1)--(vx01);
\draw[->] (vx2)--(vx02);
\draw[->] (vx3)--(vx03);
\draw[->] (vx4)--(vx04);

\draw[->] (vf1)--(vx1);
\draw[->] (vf1)--(vx2);
\draw[->] (vf2)--(vx2);
\draw[->] (vf2)--(vx3);
\draw[->] (vf3)--(vx3);
\draw[->] (vf3)--(vx4);

\draw[->] (vf1)--(vf11);
\draw[->] (vf1)--(vf12);
\draw[->] (vf1)--(vf13);

\draw[->] (vf2)--(vf21);
\draw[->] (vf2)--(vf22);
\draw[->] (vf2)--(vf23);

\draw[->] (vf3)--(vf31);
\draw[->] (vf3)--(vf32);
\draw[->] (vf3)--(vf33);

\draw[->] (vx01) --(vx11);
\draw[->] (vx01) --(vx12);
\draw[->] (vx01) --(vx13);

\draw[->] (vx02) --(vx21);
\draw[->] (vx02) --(vx22);
\draw[->] (vx02) --(vx23);

\draw[->] (vx03) --(vx31);
\draw[->] (vx03) --(vx32);
\draw[->] (vx03) --(vx33);

\draw[->] (vx04) --(vx41);
\draw[->] (vx04) --(vx42);
\draw[->] (vx04) --(vx43);

\node (empty) at (50.65, 2.25) [label=above:$ $]{};

     \end{tikzpicture}
\caption{The graph in the proof of Theorem \ref{theorem-W[2]} constructed from {\sc Hitting Set} instance $U=\{x_{1},x_{2},x_{3},x_{4}\}$, $F= \{ \{x_{1},x_{2}\},\{x_{2},x_{3}\},\{x_{3},x_{4}\}\}$ and $k=2$.}
\end{figure}

\par Let us assume that there exists a subset $S\subseteq U$ such that $|S|\leq k$ and $S \cap F \neq \emptyset$ for all $F \in \mathcal{F}$. 
We claim that every vertex in $\widetilde{G'}= G' \setminus \bigcup\limits_{x\in S} (v_x,v_x^{\prime})$ can reach at most $h$ vertices. Let us assume that there exists a 
vertex in $\widetilde{G'}$ which can reach more than $h$ vertices. 
Clearly that vertex must be $v_{F}$ for some $F\in \mathcal{F}$.
Without loss of generality assume that $x_1 \in S \cap F$. 
As we have removed the edge $(v_{x_1},v_{x_1}^{\prime})$ from $G'$, clearly $v_{F}$ cannot reach any vertex in $V_{x_1}$. 
Note that in such a case $v_{F}$ cannot reach  more than $h$ vertices as $h=n^{c}$ for some large constant. In particular, $v_{F}$ can reach  at most $h+1 - (\sum\limits_{x \in F} |V_x|) + (\sum\limits_{x \in F \setminus \{x_1\}} |V_x|) < h$ vertices.
\par In the other direction, let us assume that there exists a set $E'\subseteq E(G')$ such that $|E'|\leq k$ and every vertex in $\widetilde{G'}= G' \setminus E'$ can reach at most $h$ vertices. First we show that, given a solution $E'$ we can construct another solution $E''$ such that $E'' \subseteq \bigcup\limits_{x\in U} (v_x,v_x^{\prime})$ and $|E''|\leq |E'|$. 
To do this, we observe that the only vertices that can possibly reach more than $h$ 
vertices are  $v_{F}$. Note that if $E'$ contains an edge of the form $(v_{F},u)$ for some $u\in V_{F}$ then we can replace it by an arbitrary edge 
$(v_x,v_x^{\prime})$ for some $x\in F$. This will allow us to disconnect at least $\frac{h}{n}$ vertices from $v_{F}$ rather than just 1. 
Similar observation can be made for edges of type $(v_{F},v_x)$ for some $x\in F$ by replacing it with  edge $(v_x,v'_x)$. Therefore, we can assume that $E'' \subseteq \bigcup\limits_{x\in U}^{n} (v_x,v'_x)$. 
Next, we show that if there exists a vertex $v_{F}$ such that for every 
$x\in F$ we have $(v_x,v'_x) \not\in E''$ then 
$v_{F}$ can reach $h+1$ vertices. Clearly, $v_{F}$ can reach $V_{F}$, $\{v_x~|~ x\in F\}$ and also $\{V_x~|~ x\in F\}$. Due to construction, this set is of size more than $h$. 
This implies that for every $F\in \mathcal{F}$,  there exists an edge $(v_x,v'_x)$ for some $x\in F$ which is included in $E''$. As $|E''|\leq k$, we can define $S=\{x ~|~ (v_x,v'_x)\in E''\}$. Due to earlier observations, $S$ is a hitting set of size at most $k$.
\end{proof}

\section{ \texorpdfstring{$\mathcal{T}_{h+1}$}{T{h+1}}\textnormal{\textsc{-Free Edge Deletion}} on Split Graphs}\label{sec6}

\begin{defn}
    A graph $G$ is said to be a \emph{split graph} if $V(G)$ can be partitioned into two sets $V_1$ and $V_2$ such that $G[V_1]$ is a clique and $G[V_2]$ is an independent set.
\end{defn}

\begin{thm}\label{split} 
    The $\mathcal{T}_{h+1}$-\textnormal{\textsc{Free Edge Deletion}}  problem is NP-complete even when restricted to split graphs.
\end{thm}
\begin{proof}
     The proof is by a reduction from \textsc{Unary Bin Packing}. Let $n$ items of sizes $a_1,\dots,a_n$, and  $t$ bins of capacity $C$. 
     We construct a new instance $I'=(G,k,h)$ of $\mathcal{T}_{h+1}$-\textnormal{\textsc{Free Edge Deletion}}  as follows (see figure \ref{fig:split}). 
     Let $\alpha := n^{100}$, $k := \alpha (t-1)a + 2n(n-1) + 2n(t-1) + \binom{t}{2}$,
$h := 4(\alpha C + n + k)$, and $h' := h - (\alpha C + 2n + 1)$. Now we construct a graph $G$ in the following way:
\begin{itemize}
    \item Construct a set of vertices $X=\{v_1,\dots,v_t\}$ corresponding to the $t$ bins. 
    \item For each $j \in [t]$, add a set $P_j$ of $h'$ vertices that are adjacent to $v_j$.
    \item For each item $i \in [n]$, construct a set $A_i$ of $\alpha a_i$ vertices and a set $B_i$ of $2$ vertices. Let $A= \cup_{i=1}^n A_i$.
    \item For every $i\in [n]$, make both the vertices of $B_i$ adjacent to every vertex of $A_i$.
    \item Make each vertex of $A$ adjacent to every of $X$.
    \item Make $X  \cup (\cup_{i=1}^n B_i)$ a clique.
\end{itemize}
This construction can be done in time polynomial in the input size.
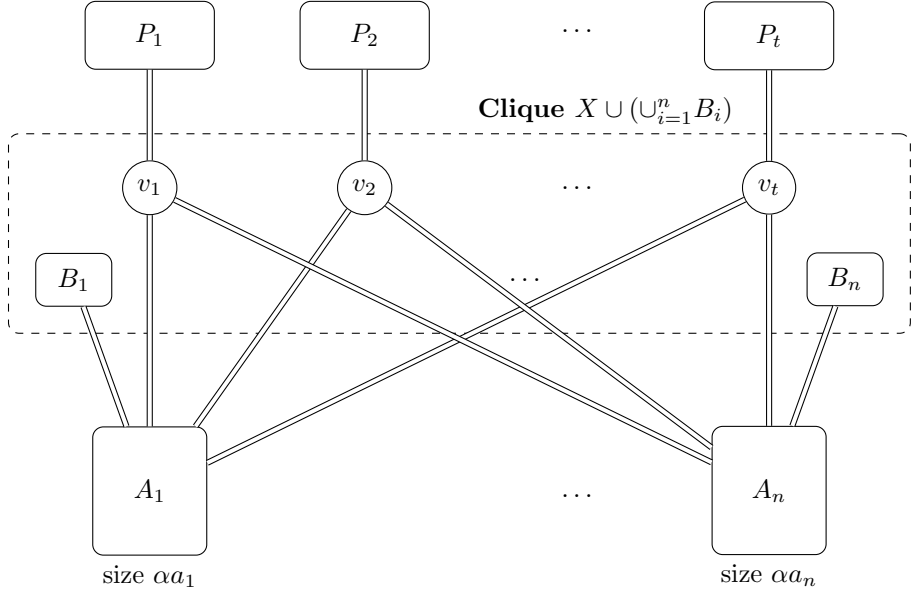
\begin{figure}[h!]
    \centering
    \begin{tikzpicture}[
    node distance=1.8cm and 2.1cm,
    v_node/.style={draw, circle, minimum size=0.7cm},
    h_node/.style={draw, rectangle, rounded corners, minimum width=1.7cm, minimum height=0.9cm},
    b_node/.style={draw, rectangle, rounded corners, minimum width=1cm, minimum height=0.7cm},
    c_node/.style={draw, rectangle, rounded corners, minimum width=1.5cm, minimum height=1.7cm},
    bundle_edge/.style={double, double distance=1.5pt, -}    
]

\node[v_node] (v1) {$v_1$};
\node[v_node] (v2) [right=of v1] {$v_2$};
\node (v_dots) [right=of v2, draw=none] {$\dots$};
\node[v_node] (vt) [right=of v2, xshift=2.5cm] {$v_t$};
\node[b_node] (B1) [below=of v1, xshift=-1cm, yshift= 1.3cm]{$B_1$};
\node (b_dots) [right=of B1, xshift=3cm,draw=none] {$\dots$};
\node[b_node] (Bn) [below=of vt, xshift=1cm, yshift= 1.3cm]{$B_n$};

\node[draw, dashed, rounded corners, fit=(v1) (v2) (v_dots) (vt) (B1) (Bn), label=85: \textbf{Clique $X \cup ( \cup_{i=1}^n B_i)$}, inner sep=10pt] (X_box) {};

\node[h_node] (H1) [above=of v1, yshift= -0.6cm] {$P_1$};
\node[h_node] (H2) [above=of v2, yshift=-0.6cm] {$P_2$};
\node (h_dots) [above=of v_dots, draw=none] {$\dots$};
\node[h_node] (Ht) [above=of vt, yshift= -0.6cm] {$P_t$};

\node[c_node] (C1) [below=of v1, yshift=-1cm, label=below:{ size $\alpha a_1$}] {$A_1$};

\node (c_dots) [below=of v_dots, draw=none, yshift=-2cm] {$\dots$};

\node[c_node] (Cn) [below=of vt, yshift=-1cm, label=below:{ size $\alpha a_n$}] {$A_n$};

\draw[bundle_edge] (v1) -- (H1);
\draw[bundle_edge] (v2) -- (H2);
\draw[bundle_edge] (vt) -- (Ht);
\draw[bundle_edge] (C1) -- (B1);
\draw[bundle_edge] (Cn) -- (Bn);

\draw[bundle_edge] (C1) -- (v1);
\draw[bundle_edge] (C1) -- (v2);
\draw[bundle_edge] (C1) -- (vt);

\draw[bundle_edge] (Cn) -- (v1);
\draw[bundle_edge] (Cn) -- (v2);
\draw[bundle_edge] (Cn) -- (vt);

\end{tikzpicture}
\caption{The graph $G$ constructed from the instance $I$ in the proof of Theorem \ref{split}. A double edge between a vertex $v$ and a set $S$ denotes that $v$ is adjacent to every vertex in $S$. A double edge between a set $A$ and a set $B$ denotes that every vertex in $A$ is adjacent to every vertex in $B$.}
\label{fig:split}
\end{figure}

 Observe that $V_1 = X \cup ( \cup_{i=1}^n B_i) $ induces a clique and $V_2 = V(G) \setminus V_1$ induces an independent set. Therefore $G$ is a split graph.
It remains to show that the instance $I'$ is equivalent to the instance $I$. Let $I$ be a yes-instance. So there exist a mapping $\gamma : [n] \rightarrow [t]$ such that $\sum_{i \in \gamma^{-1}(j)} a_i \leq C$ for every $j \in [t]$. Construct an edge set $F$ as follows:
 \begin{enumerate}[label=(\roman*)]
     \item $G[X] \subseteq F$
     \item Add all the edges between $A_i$ and $v_j$ in $F$ whenever $j \neq \gamma(i)$. 
     \item  Add all edges between $B_i$ and $v_j$ in $F$ whenever $j \neq \gamma(i)$.
     \item Add all edges between $B_i$ and $B_j$ whenever $\gamma(i) \neq \gamma(j)$.    
 \end{enumerate}
 Clearly, 
    $$|F| \leq \binom{t}{2} + \alpha a (t-1) + 2n(t-1) + 4 \binom{n}{2} = k$$
    Every component of $G\setminus F$ contains exactly one vertex of $X$. If $C_j$ is the component containing $v_j$, then 
    \begin{align*}
        |C_j| &= 1 + |P_j| +\sum_{i \in \gamma ^{-1}(j)}(|A_i| +|B_i|) \\
        &=1+ h' +\sum_{i \in \gamma ^{-1}(j)} (\alpha a_i + 2) \\
        & \leq h-\alpha C -2n +\alpha C +2n \\
        &=h 
 \end{align*}
Thus $F$ is a solution for $I'$.

Conversely, let $F \subseteq E(G)$  be a solution of $I'$.

\begin{claim}
   Each component of $G \setminus F$ contains at most one vertex of $X$. 
\end{claim}
\begin{claimproof}
    Suppose, for contradiction, there exists a component $C$ containing two vertices of $X$, say $v_i$ and $v_j$. Then \begin{align*}
    |C| & \geq  2+ |P_i| + |P_j| -k \\
    &= 2+2h' -k\\
    &=2h-2\alpha C - 4n -k\\
    &= h + h-2\alpha C - 4n -k\\
    &=h + 2 \alpha C +3k \\
    & >h.
\end{align*} 
Thus $E(G[X]) \subseteq F$. Moreover, for every $i \in [n]$, each vertex of $A_i$ and $B_i$ must have at most one $v_j$ as a neighbor in $G \setminus F$. 
This counts at least $\binom{t}{2}+ \alpha a(t-1) + 2n(t-1)$ edges in $F$. So at most $2n(n-1)$ many more edges can be there in $F$. 
\end{claimproof}

\begin{claim}
    For each $i$, both the vertices of $B_i$ will be contained in the same component in $G \setminus F$.
\end{claim}
\begin{claimproof}
    Suppose $B_i= \{b_1, b_2\}$ where $b_1$ and $b_2$ belong to two different components of $G \setminus F$. 
    Then for each vertex $a \in A_i$, either $ab_1 \in F$ or $ab_2 \in F$. 
    This demands another  $\alpha a_i$ many edges in $F$ which is not possible.
\end{claimproof}

\begin{claim}
    Each $B_i$ is connected to some $v_j$ in $G \setminus F$.
\end{claim} 
\begin{claimproof}
    Note that if this is not the case then corresponding to every vertex of $A_i$, either an edge to $B_i$ or an edge to $v_j$ should be deleted, which results in $\alpha a_i$ more edges. 
    This is also not possible since $\alpha a_i > 2n(n-1)$.
\end{claimproof}

Now we can construct a solution for $I$ as follows: For each item $i$, we assign $\gamma(i)=j$ if and only if $B_i$ and $v_j$ are in the same component of $G \setminus F$. If $|\gamma^{-1}(j)|=r$, the minimum number of vertices in the component of $G \setminus F$ containing $v_j$ will be
$$1+|P_j|+ \sum_{i \in \gamma^{-1}(j)} (|B_i|+|A_i|)  - 2n(n-1)$$
Since this is at most $h$, we have:
\begin{alignat*}{3}
&1+ h'+ 2r + \alpha  \sum_{i \in \gamma^{-1}(j)}a_i - 2n(n-1) &\leq & h\\
\implies &  h -(\alpha C+2n) + 2r + \alpha \sum_{i \in \gamma^{-1}(j)} a_i - 2n(n-1) &\leq& h \\
 \implies & -\alpha C-2n + 2r + \alpha \sum_{i \in \gamma^{-1}(j)} a_i - 2n(n-1) &\leq & 0\\
\implies &  \hspace{2.8cm} \alpha (\sum_{i \in \gamma^{-1}(j)} a_i - C) &\leq& 2n^2 - 2r < \alpha \text{ \hspace{4pt} (since } \alpha = n^{100} )\\
 \implies &  \hspace{2.8cm}\sum_{i \in \gamma^{-1}(j)} a_i - C &\leq& 0 \\
 \implies & \hspace{2.8cm} \sum_{i \in \gamma^{-1}(j)} a_i &\leq& C 
\end{alignat*}
So $\gamma$ defines a solution to $I$. Hence, $I$ and $I'$ are equivalent instances. Hence the proof.
\end{proof}
\begin{cor}
    The $\mathcal{T}_{h+1}$-\textnormal{\textsc{Free Edge Deletion}} is NP-complete even when restricted to chordal graphs.
\end{cor}
\begin{thm}
    The $\mathcal{T}_{h+1}$-\textnormal{\textsc{Free Edge Deletion}}  can be solved on split graphs in time ${\mathcal{O}((k+1)^{k+2} n^{1+o(1)})}$, where $k$ is the solution size.
\end{thm} 
\begin{proof}
Let $V(G)=(V_1, V_2)$ be the partition of the vertex set of a given split graph $G$ such that $G[V_1]$ is a clique and $G[V_2]$ is an independent set.\\
\textbf{Case 1: }When $|V_1| > k+1$\\ 
If $G,k,h)$ is a yes-instance, then $V_1$ should be in one single component in the graph obtained after deleting the solution edge set. So we assume this to be the case, and add vertices from $V_2$ greedily (in the decreasing order of degree) into this component till the component size becomes $h$. Count the number of edges between the remaining vertices of $V_2$ and $V_1$. If this number is less than or equal to $k$, return yes-instance. Otherwise, return no-instance. Time taken for this process is $\mathcal{O}(n+m)$.\\
\textbf{Case 2:} When $|V_1| \leq k+1$ \\ 
Observe that $V_1$ is a vertex cover of $G$. So the vertex cover number of $G$ is at most $k+1$. We use the algorithm by Bazgan et al.~\cite{bib11} that solves $\mathcal{T}_{h+1}$-\textnormal{\textsc{Free Edge Deletion}} in  $\mathcal{O}(\ell^{\ell+1} n^{1+o(1)})$ time where $\ell $ is the vertex cover number. Hence the problem can be solved in $\mathcal{O}((k+1)^{k+2} n^{1+o(1)})$ time.
\end{proof}

\section{ \texorpdfstring{$\mathcal{T}_{h+1}$}{T{h+1}}\textnormal{\textsc{-Free Edge Deletion}} on Interval Graphs}\label{sec7}
In this section, we provide FPT algorithm for $\mathcal{T}_{h+1}$-\textnormal{\textsc{Free Edge Deletion}} on interval graphs.

\begin{defn}
    A graph $G $ is called an \emph{interval graph} if there exists a family of intervals $\mathcal{I} = \{I_v : v \in V(G)\}$ on the real line such that for any two distinct vertices $u, v \in V(G)$, $uv \in E(G)$ if and only if $I_u \cap I_v \neq \emptyset$.
\end{defn}
The classical computational complexity of $\mathcal{T}_{h+1}$-\textnormal{\textsc{Free Edge Deletion}}  on interval graphs is currently open. In particular, it is not known whether the problem is NP-hard on this graph class when $h$ is part of the input. Nevertheless, interval graphs admit strong structural properties, most notably the existence of clique path decompositions, which can be exploited algorithmically.

In this section, we show that these structural properties suffice to obtain a fixed-parameter tractable algorithm parameterized by the solution size $k$. Our approach does not rely on resolving the classical complexity of the problem on interval graphs, but instead uses careful reasoning about the interaction between optimal solutions and the clique structure induced by a path decomposition. This yields an FPT algorithm based on structural observations that may be of independent interest.

Let $G$ be an interval graph given with a nice path decomposition $\mathcal{P}= (B_1,\dots, B_\ell)$ such that each $B_i$ induces a clique in $G$. Let $G_i= G[B_1 \cup \dots \cup  B_i]$. 
\begin{lem}\label{lem12}
   Let $(G, k, h)$ be a yes-instance and $F$ be an optimal solution. Let $F_i = F \cap E(G_i)$. Then for any bag $B_i$, at most one component of $G_i \setminus F_i$ intersecting $B_i$ has size strictly greater than $k+1$.
\end{lem}
\begin{proof}
    For the sake of contradiction, suppose there exists a bag $B_i$ containing vertices from two distinct components $C_1$ and $C_2$ of $G_i \setminus F_i$ such that $|C_1| > k+1$ and $|C_2| > k+1$.
    Let $f$ be the index of the first bag in the decomposition where vertices of both $C_1$ and $C_2$ appear together. 
    This means every node from $f$ to $i$ contains vertices of both $C_1$ and $C_2$, and the first vertex of one of the components, say $C_1$, was introduced at the node $f$. 
    Then $C_1$ must have accumulated all its vertices from the bags in the path from $f$ to $i$. 
    For each node $j \in [f,i]$ where a vertex $v \in C_1$ is introduced, an edge from $v$ to $C_2$ should be added into the solution to keep $C_1$ and $C_2$ disconnected.
    Thus $$ |C_1| \leq |F| \implies k+1 < k$$ a contradiction. Hence the proof.
\end{proof} 

\begin{thm}
      The $\mathcal{T}_{h+1}$-\textnormal{\textsc{Free Edge Deletion}}  can be solved on interval graphs in $\mathcal{O}(k^{\mathcal{O}(k)} n^3)$ time, where $k$ is the solution size.
\end{thm}

\begin{proof}
Gaikwad and Maity~\cite{bib2} showed that  $\mathcal{T}_{h+1}$-\textnormal{\textsc{Free Edge Deletion}}  admits a kernel with at most $2kh$ vertices and $2kh^2 +k$ edges. So if  $h \leq k+1$, we have a kernel of at most $2k(k+1)$ vertices and $2k(k+1)^2+k$ edges. By brute force, the problem can be solved in $\mathcal{O}(k^{\mathcal{O}(k)})$ time. Therefore, we now consider the case when $h > k+1$.

 We present a dynamic programming algorithm on the path decomposition $\mathcal{P}$ of $G$. 
We define a DP table entry $D(i,P, \alpha)$ where: 
\begin{itemize}
    \item $i$ is the index of the current node,
    \item $P=\{P_1, P_2,\dots , P_r\}$ is a partition of $B_i$,
    \item $\alpha: P \rightarrow [h]$ is a function that assigns an integer $\alpha(P_j) \in [h]$ to each part $P_j \in P$.
\end{itemize} 
$D(i,P,\alpha)$ stores the minimum size of an edge subset $F_i \subseteq E(G_i)$ such that the connected components of $G_i \setminus F_i$ intersecting $B_i$, denoted $C_1, \dots, C_r$, satisfy $C_j \cap B_i = P_j$ and $|C_j| = \alpha(P_j)$ for all $1 \leq j \leq r$. We calculate the DP table using the following recursive relations: 

 \medskip
 \noindent \textit{First node:}
 
  Since $\mathcal{P}$ is a nice path decomposition, $B_1 = \emptyset$. Hence the only value to store here is $D(1, \emptyset, \emptyset) = 0$. 
  
  \medskip
\noindent \textit{Introduce nodes:} 

Let $B_i=B_{i-1} \cup \{v\}$. Consider the partition $P$ of $B_i$ and let $v \in P_q$. Since $v$ is adjacent to every other vertex in $B_i$, all edges between $v$ and $B_i \setminus P_q$ should be in the solution. Hence, 
$$D(i,P,\alpha)= D(i-1, P',\alpha ' )+|B_i|-|P_q|$$ where $P'$ and $\alpha'$ are as follows:\\
Case 1: $|P_q| >1$.

In this case, $P' $ is obtained from $P$ by replacing $P_q$ with $P_q \setminus \{v\}$. Then $\alpha'$ is such that $\alpha'(P_j)=\alpha(P_j)$, for all $j \neq q$ and $\alpha'(P_q') = \alpha(P_q) -1$. \\
Case 2: $P_q=\{v\}$.

In this case, $P' = P \setminus \{P_q\}$. Here $\alpha'$ is the restriction of $\alpha$ to $P'$. 

 \medskip
\noindent \textit{Forget nodes:}

Let  $B_i=B_{i-1} \setminus \{v\}$, then  $$D(i,P,\alpha)= \min_{(P', \alpha')} D(i-1, P', \alpha')$$ where the minimum is taken over all pairs $(P', \alpha')$ defined by the following two cases:\\
Case 1: 

Here $P'$ is obtained from $P$ by replacing $P_q$ with $P_q \cup \{v\}$ for some $q \in [r]$. In this case, $\alpha'$ be defined by  $\alpha'(P_j')= \alpha(P_j)$ for all $j=1, \dots r$.\\
Case 2:

Here $P'= P \cup \{\{v\}\}$. In this case, $\alpha'$ is defined such that $\alpha$ is the restriction of $\alpha'$ to $P$ and $\alpha'(\{v\})=s$ for any integer $1 \leq s \leq h$.
    
Since $\mathcal{P}$ is a nice decomposition, the last node, say $\ell$, has $B_{\ell} =\emptyset$. The optimum solution will be $D(\ell, \emptyset, \emptyset)$. If this value is at most $k$, it is a yes-instance. Otherwise, return no-instance. 

To get the required complexity, we restrict the possible pairs $(P, \alpha)$ using Lemma \ref{lem12}.
For $|B_i| > k+1$, we only allow the trivial partition $P=\{B_i\}$ (since splitting the clique of size greater than $k+1$ costs more than $k$ edges). So maximum number of values to be stored at such a node is $h$.
If $|B_i| \leq k+1$, the number of possible partitions is at most $(k+1)^{k+1}$. By Lemma \ref{lem12}, at most one part $P_j$ can have $\alpha(P_j) > k+1$. So number of possible functions for a partition $P$ is at most $(k+1)^{k+1}h$. Hence the maximum number of values to be stored at such a node is $(k+1)^{2k+2}h$. 

Since the number of nodes is at most $2n$ and the computation of each value takes time $\mathcal{O}(k+h)$, the total running time is  $\mathcal{O}(k^{\mathcal{O}(k)}n^3)$.
\end{proof}

\section{Conclusion}\label{sec9}

We investigated the parameterized complexity of the $\mathcal{T}_{h+1}$-\textnormal{\textsc{Free Edge Deletion}}  problem, with particular emphasis on the case where the component size bound $h$ is unbounded. Our results provide a detailed understanding of the limits of tractability for this problem under a wide range of parameterizations.

On the hardness side, we showed that the problem is W[1]-hard when parameterized by treedepth or by twin cover, thereby strengthening and unifying earlier hardness results for treewidth, pathwidth, and feedback vertex set. Together with recent work of Bazgan, Nichterlein, and Alferez, these results indicate that several classical structural parameters do not suffice to obtain fixed-parameter tractability without bounding $h$.

On the positive side, we identified parameterizations that restore tractability. We proved that the problem is fixed-parameter tractable when parameterized by the cluster vertex deletion number together with $h$, size of vertex deletion set into a clique and when parameterized by neighborhood diversity together with $h$. In addition, since the problem is W[1]-hard when parameterized by the solution size alone, we presented a fixed-parameter tractable bicriteria approximation algorithm parameterized by $k$. We also proved that a natural generalization of this problem on directed acyclic graphs remain W[2]-hard even on directed acyclic graphs. Finally, we showed that the problem admits fixed-parameter tractable algorithms parameterized by $k$ on split graphs and on interval graphs.

Several natural questions remain open. In particular, it is unknown whether the problem admits fixed-parameter tractable algorithms on broader graph classes such as chordal graphs or planar graphs. Improving the approximation factor of our bicriteria algorithm, as well as obtaining approximation guarantees parameterized by structural parameters, are also interesting directions for future work. Furthermore, it remains open whether the problem is fixed-parameter tractable when parameterized by neighborhood diversity alone, and whether polynomial kernels exist on restricted graph classes such as split graphs or interval graphs.





\bibliography{sn-bibliography}

@InProceedings{10.1007/978-3-540-77120-3_79,
author="Guo, Jiong",
editor="Tokuyama, Takeshi",
title="Problem Kernels for NP-Complete Edge Deletion Problems: Split and Related Graphs",
booktitle="Algorithms and Computation",
year="2007",
publisher="Springer Berlin Heidelberg",
address="Berlin, Heidelberg",
pages="915--926",
isbn="978-3-540-77120-3"
}

@article{1e396962aa954ff5bbccabe84d44e71c,
title = "A Survey of Component Order Connectivity Models of Graph Theoretic Networks",
author = "Daniel Gross and Monika Heinig and Lakshmi Iswara and Kazmierczak, \{L. William\} and Kristi Luttrell and John Saccoman and Charles Suffel",
year = "2013",
month = sep,
language = "American English",
volume = "12",
pages = "895--910",
journal = "WSEAS Transactions on Mathematics",
number = "9",
}

@article{drange2016vertex,
  author  = {Drange, Pål Grønås and Dregi, Markus and van ’t Hof, Pim},
  title   = {On the Computational Complexity of Vertex Integrity and Component Order Connectivity},
  journal = {Algorithmica},
  volume  = {76},
  number  = {4},
  pages   = {1181--1202},
  year    = {2016},
  doi     = {10.1007/s00453-016-0127-x}
}

@article{CAI1996171,
title = {Fixed-parameter tractability of graph modification problems for hereditary properties},
journal = {Information Processing Letters},
volume = {58},
number = {4},
pages = {171-176},
year = {1996},
issn = {0020-0190},
doi = {https://doi.org/10.1016/0020-0190(96)00050-6},
url = {https://www.sciencedirect.com/science/article/pii/0020019096000506},
author = {Leizhen Cai}
}

@article{Gross2013ComponentOrder,
  author  = {Gross, Donald and Heinig, Mark and Iswara, Lakshmanan and Kazmierczak, Lisa W. and Luttrell, Kevin and Saccoman, John T. and Suffel, Cliff},
  title   = {A survey of component order connectivity models of graph theoretic networks},
  journal = {SIAM Journal on Applied Mathematics},
  volume  = {12},
  number  = {9},
  pages   = {1--26},
  year    = {2013}
}

@InProceedings{10.1007/978-3-642-20877-5_30,
author="Li, Angsheng
and Tang, Linqing",
editor="Ogihara, Mitsunori
and Tarui, Jun",
title="The Complexity and Approximability of Minimum Contamination Problems",
booktitle="Theory and Applications of Models of Computation",
year="2011",
publisher="Springer Berlin Heidelberg",
address="Berlin, Heidelberg",
pages="298--307",
isbn="978-3-642-20877-5"
}

@article{bib1,
  author		= "{Enright, Jessica} and {Meeks, Kitty}",
  title			= "Deleting Edges to Restrict the Size of an Epidemic: A New Application for Treewidth",
  journal		= "Algorithmica",
  volume		= "80",
  pages			= "1857--1889",
  year			= "2018",
  doi        = "10.1007/s00453-017-0311-7"
}

@article{bib2,
title = {Further parameterized algorithms for the F-free edge deletion problem},
journal = {Theoretical Computer Science},
volume = {933},
pages = {125-137},
year = {2022},
issn = {0304-3975},
doi = {https://doi.org/10.1016/j.tcs.2022.08.025},
url = {https://www.sciencedirect.com/science/article/pii/S0304397522005205},
author = {Ajinkya Gaikwad and Soumen Maity},
}

@InProceedings{bib4,
author="Fellows, Michael R.
and Lokshtanov, Daniel
and Misra, Neeldhara
and Rosamond, Frances A.
and Saurabh, Saket",
editor="Hong, Seok-Hee
and Nagamochi, Hiroshi
and Fukunaga, Takuro",
title="Graph Layout Problems Parameterized by Vertex Cover",
booktitle="Algorithms and Computation",
year="2008",
publisher="Springer Berlin Heidelberg",
address="Berlin, Heidelberg",
pages="294--305",
isbn="978-3-540-92182-0",
doi= "10.1007/978-3-540-92182-0_28"
}

@article{bib5,
title = {Bin packing with fixed number of bins revisited},
journal = {Journal of Computer and System Sciences},
volume = {79},
number = {1},
pages = {39-49},
year = {2013},
issn = {0022-0000},
doi = {https://doi.org/10.1016/j.jcss.2012.04.004},
url = {https://www.sciencedirect.com/science/article/pii/S0022000012000943},
author = {Klaus Jansen and Stefan Kratsch and Dániel Marx and Ildikó Schlotter},
}

@article{bib8,
  title={Minimum bisection is fixed parameter tractable},
  author={Marek Cygan and Daniel Lokshtanov and Marcin Pilipczuk and Michał Pilipczuk and Saket Saurabh},
  journal={Proceedings of the forty-sixth annual ACM symposium on Theory of computing},
  year={2013},
  url={https://api.semanticscholar.org/CorpusID:675985}
}

@article{bib9,
author = {Boral, Anudhyan and Cygan, Marek and Kociumaka, Tomasz
 and Pilipczuk, Marcin},
year = {2016},
pages = {357-376},
title = {A Fast Branching Algorithm for Cluster Vertex Deletion},
volume = {58},
journal = {Theory of Computing Systems},
doi = {10.1007/s00224-015-9631-7}
}

@article{Neil,
	Author = {Neil Robertson and P.D Seymour},
	Journal = {Journal of Combinatorial Theory, Series B},
	Number = {1},
	Pages = {49 - 64},
	Title = {Graph minors. III. Planar tree-width},
	Volume = {36},
	Year = {1984}}

@ARTICLE{2013arXiv1306.3181C,
       author = {{Cao}, Yixin},
        title = "{An Efficient Branching Algorithm for Interval Completion}",
      journal = {arXiv e-prints},
         year = 2013,
        month = jun,
          eid = {arXiv:1306.3181},
        pages = {arXiv:1306.3181},
          doi = {10.48550/arXiv.1306.3181},
archivePrefix = {arXiv},
       eprint = {1306.3181},
 primaryClass = {cs.DS},
       adsurl = {https://ui.adsabs.harvard.edu/abs/2013arXiv1306.3181C}
}

@misc{lokshtanov2017parameterizedintegerquadraticprogramming,
      title={Parameterized Integer Quadratic Programming: Variables and Coefficients}, 
      author={Daniel Lokshtanov},
      year={2017},
      eprint={1511.00310},
      archivePrefix={arXiv},
      primaryClass={cs.DS},
      url={https://arxiv.org/abs/1511.00310}, 
}

@article{CRESPELLE2023100556,
title = {A survey of parameterized algorithms and the complexity of edge modification},
journal = {Computer Science Review},
volume = {48},
pages = {100556},
year = {2023},
issn = {1574-0137},
doi = {https://doi.org/10.1016/j.cosrev.2023.100556},
url = {https://www.sciencedirect.com/science/article/pii/S1574013723000230},
author = {Christophe Crespelle and Pål Grønås Drange and Fedor V. Fomin and Petr Golovach}
}

@misc{gaikwad2025parameterizedalgorithmseditinguniform,
      title={Parameterized Algorithms for Editing to Uniform Cluster Graph}, 
      author={Ajinkya Gaikwad and Hitendra Kumar and Soumen Maity},
      year={2025},
      eprint={2404.10023},
      archivePrefix={arXiv},
      primaryClass={cs.DS},
      url={https://arxiv.org/abs/2404.10023}, 
}

@InProceedings{misra_et_al:LIPIcs.ISAAC.2023.53,
  author =	{Misra, Neeldhara and Mittal, Harshil and Saurabh, Saket and Thakkar, Dhara},
  title =	{{On the Complexity of the Eigenvalue Deletion Problem}},
  booktitle =	{34th International Symposium on Algorithms and Computation (ISAAC 2023)},
  pages =	{53:1--53:17},
  series =	{Leibniz International Proceedings in Informatics (LIPIcs)},
  ISBN =	{978-3-95977-289-1},
  ISSN =	{1868-8969},
  year =	{2023},
  volume =	{283},
  editor =	{Iwata, Satoru and Kakimura, Naonori},
  publisher =	{Schloss Dagstuhl -- Leibniz-Zentrum f{\"u}r Informatik},
  address =	{Dagstuhl, Germany},
  URL =		{https://drops.dagstuhl.de/entities/document/10.4230/LIPIcs.ISAAC.2023.53},
  URN =		{urn:nbn:de:0030-drops-193555},
  doi =		{10.4230/LIPIcs.ISAAC.2023.53}
}

@book{marekcygan,
  author    = {Marek Cygan and
               Fedor V. Fomin and
               Lukasz Kowalik and
               Daniel Lokshtanov and
               D{\'{a}}niel Marx and
               Marcin Pilipczuk and
               Michal Pilipczuk and
               Saket Saurabh},
  title     = {Parameterized Algorithms},
  publisher = {Springer},
  year      = {2015}
 }

@article{CaoMarx2016ChordalEditing,
  author  = {Cao, Yixin and Marx, D{\'a}niel},
  title   = {Chordal Editing is Fixed-Parameter Tractable},
  journal = {Algorithmica},
  volume  = {75},
  number  = {1},
  pages   = {118--137},
  year    = {2016},
  doi     = {10.1007/s00453-015-0014-x}
}

@article{MATHIESON20103181,
title = {The parameterized complexity of editing graphs for bounded degeneracy},
journal = {Theoretical Computer Science},
volume = {411},
number = {34},
pages = {3181-3187},
year = {2010},
issn = {0304-3975},
doi = {https://doi.org/10.1016/j.tcs.2010.05.015},
url = {https://www.sciencedirect.com/science/article/pii/S0304397510002963},
author = {Luke Mathieson}
}

@misc{gaikwad2025parameterizedcomplexitysclubcluster,
      title={Parameterized Complexity of s-Club Cluster Edge Deletion: When Is the Diameter Bound Necessary?}, 
      author={Ajinkya Gaikwad},
      year={2025},
      eprint={2510.07065},
      archivePrefix={arXiv},
      primaryClass={cs.DM},
      url={https://arxiv.org/abs/2510.07065}, 
}

@article{Lampis,
title={Algorithmic Meta-theorems for Restrictions of Treewidth},
author={Lampis, M.},
journal={Algorithmica},
year={2012},
volume={64},
pages={19-37}
}

@book{13638,
	author = {Diestel, Reinhard},
	title = {Graph theory /},
	publisher = {Springer Nature,},
	year = {2017.},
	series = {Graduate Texts in Mathematics; 173},
	address = {Germany :},
	edition = {5th ed.}
}

@book{Downey,
 author = {Downey, Rodney G. and Fellows, M. R.},
 title = {Parameterized Complexity},
 year = {2012},
 publisher = {Springer},
}

@article{defmodwidth,
title = {Modular decomposition and transitive orientation},
journal = {Discrete Mathematics},
volume = {201},
number = {1},
pages = {189-241},
year = {1999},
issn = {0012-365X},
doi = {https://doi.org/10.1016/S0012-365X(98)00319-7},
url = {https://www.sciencedirect.com/science/article/pii/S0012365X98003197},
author = {Ross M. McConnell and Jeremy P. Spinrad}
}

@inproceedings{bib10,
author = {Tedder, Marc and Corneil, Derek and Habib, Michel and Paul, Christophe},
title = {Simpler Linear-Time Modular Decomposition Via Recursive Factorizing Permutations},
year = {2008},
isbn = {3540705740},
publisher = {Springer-Verlag},
address = {Berlin, Heidelberg},
doi = {10.1007/978-3-540-70575-8_52},
pages = {634–645},
numpages = {12},
location = {Reykjavik, Iceland},
series = {ICALP '08},
}

@article{bib11,
author = {Bazgan, Cristina and Nichterlein, Andr\'{e} and Vazquez Alferez, Sofia},
title = {Destroying densest subgraphs is hard},
year = {2025},
issue_date = {Aug 2025},
publisher = {Academic Press, Inc.},
address = {USA},
volume = {151},
number = {C},
issn = {0022-0000},
url = {https://doi.org/10.1016/j.jcss.2025.103635},
doi = {10.1016/j.jcss.2025.103635},
journal = {J. Comput. Syst. Sci.},
month = aug,
numpages = {22}
}

@article{bib12,
title = {Tree-depth, subgraph coloring and homomorphism bounds},
journal = {European Journal of Combinatorics},
volume = {27},
number = {6},
pages = {1022-1041},
year = {2006},
issn = {0195-6698},
doi = {https://doi.org/10.1016/j.ejc.2005.01.010},
url = {https://www.sciencedirect.com/science/article/pii/S0195669805000570},
author = {Jaroslav Nešetřil and Patrice {Ossona de Mendez}}
}

@article{bib13,
author = {Ganian, Robert},
year = {2015},
month = {09},
pages = {77-100},
title = {Improving Vertex Cover as a Graph Parameter},
volume = {Vol. 17 no.2},
journal = {Discrete Mathematics \& Theoretical Computer Science},
doi = {10.46298/dmtcs.2136}
}

@article{bib14,
title = {Upper bounds to the clique width of graphs},
journal = {Discrete Applied Mathematics},
volume = {101},
number = {1},
pages = {77-114},
year = {2000},
issn = {0166-218X},
doi = {https://doi.org/10.1016/S0166-218X(99)00184-5},
url = {https://www.sciencedirect.com/science/article/pii/S0166218X99001845},
author = {Bruno Courcelle and Stephan Olariu}
}

\end{document}